\RequirePackage{fix-cm}
\documentclass[natbib,smallcondensed]{svjour3}     
\PassOptionsToPackage{numbers,sort&compress}{natbib}
\bibpunct[,]{[}{]}{,}{n}{,}{,}
\smartqed  
\usepackage{graphicx,amssymb,amsmath,nicefrac,verbatim,colonequals,tikz,enumerate}

\usepackage[ruled,vlined,linesnumbered]{algorithm2e}
\SetKwBlock{Repeat}{repeat}{}

\usetikzlibrary{decorations.pathreplacing, patterns}
\newlength{\hatchspread}
\newlength{\hatchthickness}
\newlength{\hatchshift}
\newcommand{\hatchcolor}{}
\tikzset{hatchspread/.code={\setlength{\hatchspread}{#1}},
	hatchthickness/.code={\setlength{\hatchthickness}{#1}},
	hatchshift/.code={\setlength{\hatchshift}{#1}},
	hatchcolor/.code={\renewcommand{\hatchcolor}{#1}}}
\tikzset{hatchspread=3pt,
	hatchthickness=0.3pt,
	hatchshift=0pt,
	hatchcolor=black}
\pgfdeclarepatternformonly[\hatchspread,\hatchthickness,\hatchshift,\hatchcolor]
{my north west lines}
{\pgfqpoint{\dimexpr-2\hatchthickness}{\dimexpr-2\hatchthickness}}
{\pgfqpoint{\dimexpr\hatchspread+2\hatchthickness}{\dimexpr\hatchspread+2\hatchthickness}}
{\pgfqpoint{\dimexpr\hatchspread}{\dimexpr\hatchspread}}
{
	\pgfsetlinewidth{\hatchthickness}
	\pgfpathmoveto{\pgfqpoint{0pt}{\dimexpr\hatchspread+\hatchshift}}
	\pgfpathlineto{\pgfqpoint{\dimexpr\hatchspread+0.15pt+\hatchshift}{-0.15pt}}
	\ifdim \hatchshift > 0pt
	\pgfpathmoveto{\pgfqpoint{0pt}{\hatchshift}}
	\pgfpathlineto{\pgfqpoint{\dimexpr0.15pt+\hatchshift}{-0.15pt}}
	\fi
	\pgfsetstrokecolor{\hatchcolor}
	\pgfusepath{stroke}
}

\newtheorem{assumption}{Assumption}

\newcommand{\Q}{\mathbb{Q}}
\newcommand{\N}{\mathbb{N}}

\newcommand{\gap}{\textsc{Gap}}
\newcommand{\constrained}{\textsc{Constrained}}
\newcommand{\restrict}{\textsc{Restrict}}

\newcommand{\dualrestrict}{\textsc{DualRestrict}}
\newcommand{\shortestpath}{\textsc{ShortestPath}}
\newcommand{\spanningtree}{\textsc{SpanningTree}}
\newcommand{\partition}{\textsc{Partition}}

\newcommand{\opt}{\textsc{opt}}
\newcommand{\xnext}{x^{\textnormal{next}}}

\newcommand{\xleft}{x^{\textnormal{left}}}
\begin{document}

\title{One-Exact Approximate Pareto Sets \thanks{This work was supported by the bilateral cooperation project ``Approximation methods for multiobjective optimization problems'' funded by the German Academic Exchange Service (DAAD, Project-ID 57388848) and by Campus France, PHC PROCOPE 2018 (Project no. 40407WF) as well as by the DFG grants RU 1524/6-1 and TH 1852/4-1.}
}


\author{Arne Herzel \and
        Cristina Bazgan         \and
        Stefan Ruzika			\and
        Clemens Thielen			\and
        Daniel Vanderpooten
}

\authorrunning{A.~Herzel, C.~Bazgan, S.~Ruzika, C.~Thielen, D.~Vanderpooten} 

\institute{Cristina Bazgan \and Daniel Vanderpooten \at
			Universit\'e Paris-Dauphine, PSL Research University, CNRS,  LAMSADE, 75016 Paris, France, \email{\{bazgan,daniel.vanderpooten\}@lamsade.dauphine.fr} \\
			\and
			Arne Herzel \and Stefan Ruzika \and Clemens Thielen \at
			University of Kaiserslautern, Department of Mathematics,
			Paul-Ehrlich-Str.~14,\\ D-67663~Kaiserslautern, Germany,
			\email{\{herzel,ruzika,thielen\}@mathematik.uni-kl.de}
}


\maketitle

\begin{abstract}
  Papadimitriou and Yannakakis~\cite{Papadimitriou+Yannakakis:multicrit-approx} show that the polynomial-time solvability of a certain singleobjective problem determines the class of multiobjective optimization problems that admit a polynomial-time computable $(1+\varepsilon, \dots , 1+\varepsilon)$-approximate Pareto set (also called an \emph{$\varepsilon$-Pareto set}). Similarly, in this article, we characterize the class of problems having a polynomial-time computable approximate $\varepsilon$-Pareto set that is \emph{exact in one objective} by the efficient solvability of an appropriate singleobjective problem. This class includes important problems such as multiobjective shortest path and spanning tree, and the approximation guarantee we provide is, in general, best possible. Furthermore, for biobjective problems from this class, we provide an algorithm that computes a one-exact $\varepsilon$-Pareto set of cardinality at most twice the cardinality of a smallest such set and show that this factor of~$2$ is best possible. For three or more objective functions, however, we prove that no constant-factor approximation on the size of the set can be obtained efficiently.
\end{abstract}

\newpage

\section{Introduction}\label{sec:introduction}

In many cases, real-world optimization problems involve several conflicting objectives, e.g., the minimization of cost and time in transportation systems or the maximization of profit and security in investments. In this context, solutions optimizing all objectives simultaneously usually do not exist. Therefore, in order to support decision making, so-called \emph{efficient} (or \emph{Pareto optimal}) solutions achieving a good compromise among the objectives are considered. More formally, a solution is said to be efficient if any other solution that is better in some objective is necessarily worse in at least one other objective. The image of an efficient solution in the objective space is  called a \emph{nondominated point}. 

When no prior preference information is available, one main goal of multiobjective optimization is to determine the set of all nondominated points and provide, for each of them, one corresponding efficient solution. 

\smallskip

Several results in the literature, however, show that multiobjective optimization problems are hard to solve exactly~\cite{Serafini:Complexity,Ehrgott:book} and, in addition, the cardinalities of the set of nondominated points (the \emph{nondominated set}) and the set of efficient solutions (the \emph{efficient set}) may be exponentially large for discrete problems (and are typically infinite for continuous problems). 

This impairs the applicability of exact solution methods to real-life problems and provides a strong motivation for studying \emph{approximations of multiobjective optimization problems}.

\smallskip


\subsection{Related Work}\label{subsec:related-work}

The systematic study of generally applicable approximation methods for multiobjective optimization problems started with the seminal work of Papadimitriou and Yannakakis~\cite{Papadimitriou+Yannakakis:multicrit-approx}. They show that, for any $\varepsilon>0$, any multiobjective optimization problem with a constant number of positive-valued, polynomially computable objective functions admits a $(1+\varepsilon,\dots,1+\varepsilon)$-approximate Pareto set (also called an \emph{$\varepsilon$-Pareto set}) with cardinality polynomial in the encoding length of the input and $\frac{1}{\varepsilon}$. Moreover, they show that such a set is computable in (fully) polynomial time if and only if the following auxiliary problem called the \emph{gap problem} ($\gap$) can be solved in (fully) polynomial time:\footnote{The definition of $\gap$ provided here is for minimization problems. The definition for maximization problems is completely analogous.}

	Given an instance of a $p$-objective minimization problem, a vector~$b\in \mathbb{R}^p$, and $\delta>0$, either return a feasible solution~$x$ whose objective value $f(x)\in\mathbb{R}^p$ satisfies $f_j(x)\leq b_j$ for all~$j$ or answer correctly that there is no feasible solution~$x'$ with $f_j(x')\leq \frac{b_j}{1+\delta}$ for all~$j$.

Thus, the result of Papadimitriou and Yannakakis~\cite{Papadimitriou+Yannakakis:multicrit-approx} shows that the po\-ly\-no\-mi\-al-time solvability of $\gap$ provides a complete characterization of the class of problems for which $\varepsilon$-Pareto sets can be computed in polynomial time.

\smallskip

More recent articles building upon the results of~\cite{Papadimitriou+Yannakakis:multicrit-approx} present methods that additionally yield bounds on the cardinality of the computed $\varepsilon$-Pareto set relative to the cardinality of a smallest $\varepsilon$-Pareto set possible~\cite{Vassilvitskii+Yannakakis:trade-off-curves,Diakonikolas+Yannakakis:approx-pareto-sets,Bazgan+etal:min-pareto,Koltun+Papadimitriou:approx-dom-repr}.

\smallskip

Koltun and Papadimitriou~\cite{Koltun+Papadimitriou:approx-dom-repr} show that, if all feasible solutions of a biobjective problem are given explicitly in the input (which is usually not the case for combinatorial problems, where the feasible set is in most cases given implicitly, and its cardinality is exponentially large in the input size), it is possible to compute an $\varepsilon$-Pareto set of minimum cardinality in polynomial time using a greedy procedure. This greedy procedure can be generalized to the case that the budget-constrained problem associated with the given biobjective problem can be solved exactly in polynomial time~\cite{Diakonikolas+Yannakakis:approx-pareto-sets}. For three or more objectives, however, computing a minimum-cardinality $\varepsilon$-Pareto set is $\textsf{NP}$-hard even if all feasible solutions are given explicitly.

\smallskip

Again for biobjective problems, Vassilvitskii and Yannakakis~\cite{Vassilvitskii+Yannakakis:trade-off-curves} show that, using a polynomial-time algorithm for $\gap$ as a subroutine, it is possible to compute an $\varepsilon$-Pareto set whose cardinality is at most $3$~times larger than the cardinality of a smallest $\varepsilon$-Pareto set. Moreover, this factor of~$3$ is shown to be best possible in two different ways: (1) No \emph{generic} algorithm that uses only a routine for $\gap$ can obtain a factor smaller than~$3$ without solving $\gap$ for exponentially large values of~$\frac{1}{\delta}$ (even if $\textsf{P}=\textsf{NP}$), and (2) for some biobjective problems for which $\gap$ is polynomially solvable, it is $\textsf{NP}$-hard to obtain a factor smaller than~$3$. An alternative, simpler algorithm based on $\gap$ that also obtains a factor of~$3$ is presented in~\cite{Bazgan+etal:min-pareto}. For three or more objectives, however, Vassilvitskii and Yannakakis~\cite{Vassilvitskii+Yannakakis:trade-off-curves} show that no generic algorithm based on solving $\gap$ can obtain any constant factor with respect to the cardinality of a smallest $\varepsilon$-Pareto set.

\smallskip

Diakonikolas and Yannakakis~\cite{Diakonikolas+Yannakakis:approx-pareto-sets} show that, for a broad class of biobjective problems including $\shortestpath$ and $\spanningtree$, a factor of~$2$ can be obtained with respect to the cardinality of a smallest $\varepsilon$-Pareto set. To achieve this, they use subroutines for two different 
singleobjective auxiliary problems called $\restrict$ and $\dualrestrict$. Both of these problems are harder to solve than $\gap$, but $\restrict$ and $\dualrestrict$ are polynomially equivalent to each other for biobjective problems. The factor of~$2$ is again shown to be best possible in~\cite{Diakonikolas+Yannakakis:approx-pareto-sets} in the sense that no generic algorithm based on $\restrict$ and $\dualrestrict$ can obtain a smaller factor and, for some biobjective problems for which $\restrict$ and $\dualrestrict$ are polynomially solvable, it is $\textsf{NP}$-hard to obtain a smaller factor. For three or more objectives, however, it is not known whether $\restrict$ and $\dualrestrict$ can be used to improve upon the results obtained via $\gap$ with respect to the computation of (small) $\varepsilon$-Pareto sets. Moreover, these two subproblems are not polynomially equivalent anymore in the case of three or more objectives.

\smallskip

There are also many specialized approximation algorithms for particular multiobjective problems available. Among those, there are two algorithms that actually yield approximations that are exact in one objective: For multiobjective $\shortestpath$, Tsaggouris and Zaroliagis~\cite{Tsaggouris+Zaroliagis:mult-shortest-path} present a dynamic-programming-based algorithm that yields a $(1,1+\varepsilon,\dots,1+\varepsilon)$-approximate Pareto set for any number of objective functions. For the min-cost-makespan scheduling problem, Angel et al.~\cite{Angel+etal:min-cost-makespan,Angel+etal:bicriteria-scheduling} present an algorithm computing a $(1,1+\varepsilon)$-approximate Pareto set. Neither of these algorithms, however, yields any worst-case guarantee on the cardinality of the computed approximate Pareto set.

\subsection{Our Contribution} \label{subsec:our-contribution}

We consider general multiobjective optimization problems with an arbitrary, fixed number of objectives and show that, for any such problem, there exist poly\-nomially-sized $\varepsilon$-Pareto sets that are \emph{exact in one objective function}. Assuming without loss of generality that the first objective function is the one to be optimized exactly, we refer to such $(1,1+\varepsilon,\dots,1+\varepsilon)$-approximate Pareto sets as \emph{one-exact} $\varepsilon$-Pareto sets. Consequently, we improve upon the existence result for polynomially-sized $\varepsilon$-Pareto sets of Papadimitriou and Yannakakis~\cite{Papadimitriou+Yannakakis:multicrit-approx} using the same assumptions. Moreover, the approximation guarantee of $(1,1+\varepsilon,\dots,1+\varepsilon)$ we provide is best possible in the sense that polynomially-sized approximate Pareto sets that are exact in more than one objective do, in general, not exist.

\smallskip

We then consider the class of problems for which the $\dualrestrict$ subproblem considered in~\cite{Diakonikolas+Yannakakis:approx-pareto-sets} can be solved in polynomial time. We first show that, for any constant number of objective functions, the polynomial-time solvability of $\dualrestrict$ characterizes the class of problems for which one-exact $\varepsilon$-Pareto sets can be computed in polynomial time. Consequently, even for more than two objective functions, our result provides a complete characterization of the approximation quality achievable for the class of problems studied in~\cite{Diakonikolas+Yannakakis:approx-pareto-sets}, which includes $\shortestpath$, $\spanningtree$, and many more.

\smallskip

Moreover, we provide results about the cardinality of the computed one-exact $\varepsilon$-Pareto sets compared to the cardinality of a smallest such set. While we show that the cardinality of a smallest one-exact $\varepsilon$-Pareto set (i.e., a one-exact $\varepsilon$-Pareto set having minimum cardinality among all one-exact $\varepsilon$-Pareto sets) can be much larger than the cardinality of a smallest $\varepsilon$-Pareto set (i.e., an $\varepsilon$-Pareto set having minimum cardinality among all $\varepsilon$-Pareto sets) even for biobjective problems, we also prove that the cardinality of a smallest one-exact $\varepsilon$-Pareto set can again be approximated up to a factor of~$2$ in the biobjective case by using a generic algorithm based on solving $\dualrestrict$, and we show that this factor is best possible given our assumptions. For three or more objectives, however, it is again impossible to obtain any constant factor approximation on the cardinality efficiently by using only routines for $\dualrestrict$.

\smallskip

For multiobjective $\spanningtree$, our generic algorithms yield the first poly\-nomial-time methods for computing one-exact $\varepsilon$-Pareto sets when using the algorithm provided in~\cite{Grandoni+etal:new-approaches} to solve $\dualrestrict$. For multiobjective $\shortestpath$, using the algorithm from~\cite{Horvath+Kis:Dual-RSPP} to solve $\dualrestrict$, our algorithms have running times competitive with the running time of the specialized algorithm for computing one-exact $\varepsilon$-Pareto sets for $\shortestpath$ provided in~\cite{Tsaggouris+Zaroliagis:mult-shortest-path}. This is particularly noteworthy since the algorithm from~\cite{Tsaggouris+Zaroliagis:mult-shortest-path} is currently the algorithm with the best worst-case running time for computing $\varepsilon$-Pareto sets for $\shortestpath$ even when no objective function is to be optimized exactly. Moreover, for the case of two objectives, our biobjective algorithm additionally provides a worst-case guarantee on the cardinality of the computed one-exact $\varepsilon$-Pareto set (while the algorithm from~\cite{Tsaggouris+Zaroliagis:mult-shortest-path} provides no such guarantee).

	\section{Preliminaries}\label{sec:preliminaries}

We consider general multiobjective minimization and maximization problems formally defined as follows:

\begin{definition} [Multiobjective Minimization/Maximization Problem]
	A \emph{multiobjective optimization problem}~$\Pi$ is given by a set of instances. Each instance $I$ consists of a (finite or infinite) set~$X^I$ of feasible solutions and a vector~$f^I = (f^I_1,\ldots, f^I_{p})$ of $p$ objective functions~$f^I_i : X^I \to \Q$ for $i = 1\ldots,p$. In a minimization problem, all objective functions~$f^I_i$ should be minimized, in a maximization problem, they should be maximized. The feasible set~$X^I$ might not be given explicitly. 
\end{definition}

Here, the number~$p$ of objective functions in a multiobjective optimization problem~$\Pi$ is assumed to be constant. The solutions of interest are those for which it is not possible to improve the value of one objective function without worsening the value of at least one other objective. Solutions with this property are called \emph{efficient solutions}:

\begin{definition}
	For an instance~$I$ of a minimization (maximization) problem, a solution~$x \in X^I$ \emph{dominates} another solution~$x' \in X^I$ if $f^I_i(x) \leq f^I_i(x')$ ($f^I_i(x) \geq f^I_i(x')$) for $i = 1,\ldots, p$ and $f^I_i(x) < f^I_i(x')$ ($f^I_i(x) > f^I_i(x')$) for at least one~$i$. A solution~$x \in X^I$ is called \emph{efficient} if it is not dominated by any other solution~$x' \in X^I$. In this case, we call the corresponding image~$f^I(x) \in f^I(X^I) \subseteq \Q^p$ a \emph{nondominated point}. The set~$X^I_E\subseteq X^I$ of all efficient solutions is called the \emph{efficient set} (or \emph{Pareto set}) and the set~$Y^I_N\colonequals f^I(X^I_E)$ of nondominated points is called the \emph{nondominated set}.
\end{definition}

The goal in multiobjective optimization typically consists of computing the nondominated set~$Y^I_N$ and, for each nondominated point~$y\in Y^I_N$, one corresponding efficient solution~$x\in X^I_E$ with $f(x)=y$.

\smallskip

Throughout the paper, we make the standard assumption of rational, positive-valued, polynomially computable objective functions used in the context of approximation of multiobjective problems (cf.~\cite{Papadimitriou+Yannakakis:multicrit-approx,Vassilvitskii+Yannakakis:trade-off-curves,Diakonikolas+Yannakakis:approx-pareto-sets}), which is formalized in Assumption~\ref{assum:M}.


\begin{assumption} \label{assum:M}
	For any multiobjective optimization problem~$\Pi$, there exists a polynomial~$\mathcal{P}$ such that, for any instance~$I$ of~$\Pi$, there exists a constant~$M^I \leq \mathcal{P}(\textnormal{enc}(I))$ such that $\textnormal{enc}(f^I_i(x)) \leq M^I$ for any $x \in X^I$ and any $i \in \{1,\ldots, p\}$, where $\textnormal{enc}(I)$ denotes the encoding length of instance~$I$ and $\textnormal{enc}(f^I_i(x))$ denotes the encoding length of the value~$f^I_i(x) \in \Q_{>0}$ in binary. This, in particular, implies that, for any instance~$I$ and any objective function value~$f_i^I(x)$, we have $2^{-M^I} \leq f_i^I(x) \leq 2^{M^I}$. Also, any two values~$f_i^I(x)$ and $f_i^I(x')$ differ by at least~$2^{-2M^I}$ if they are not equal.
\end{assumption}


In the following, we blur the distinction between the problem~$\Pi$ and a concrete instance $I = (X^I,f^I)$ and usually drop the superscript~$I$ indicating the dependence on the instance in~$X^I$, $f^I$, $M^I$, etc. 

\smallskip

Multiobjective optimization problems consist of objective functions that are to be minimized or objectives that are to be maximized (or even a combination of both). However, we only consider minimization objectives in this article. This is without loss of generality here since all our arguments can be straightforwardly adapted to maximization problems.

\smallskip

One of the main issues in multiobjective optimization problems is that the nondominated set often consists of exponentially many points, which renders the problem intractable (see, e.g.,~\cite{Ehrgott:hard-to-say}). One way to overcome this obstacle is the concept of approximation. 
Instead of computing at least one corresponding efficient solution for each point in the nondominated set, we only require each image point in the objective space to be ``almost'' dominated by the image of a solution from the computed set.



\begin{definition}[Approximate Pareto Set]
	Let $(X,f)$ be a multiobjective optimization problem and let $\alpha_i \geq 1$ for $i = 1,\ldots,p$. We say that a feasible solution~$x \in X$ $(\alpha_1,\ldots,\alpha_p)$-\emph{approximates} another feasible solution~$x' \in X$ if $f(x)$ $(\alpha_1,\ldots,\alpha_p)$-\emph{dominates} $f(x')$, i.e., for minimization problems, if $f_i(x) \leq \alpha_i \cdot f_i(x')$ and, for maximization problems, if  $ \alpha_i \cdot f_i(x) \geq f_i(x')$, for all $i = 1,\ldots,p$. A set~$P_{(\alpha_1,\ldots,\alpha_p)} \subseteq X$ of feasible solutions is called an $(\alpha_1,\ldots,\alpha_p)$\emph{-approximate Pareto set} if, for every feasible solution $x' \in X$, there exists a solution $x \in P_{(\alpha_1,\ldots,\alpha_p)}$ that $(\alpha_1,\ldots,\alpha_p)$-approximates $x'$. For $\varepsilon > 0$, a $(1+\varepsilon,\ldots,1+\varepsilon)$-approximate Pareto set is called an $\varepsilon$\emph{-Pareto set}.
\end{definition}

\begin{remark}\label{rmk:no-two-ones}
	If $\alpha_i = 1$ for two or more indices~$i$, there exist problems for which any $(\alpha_1,\ldots,\alpha_p)$-approximate Pareto set requires exponentially many solutions. This follows since even many biobjective optimization problems (e.g., the biobjective $\shortestpath$ problem) admit instances with exponentially many different nondominated points (see, e.g.,~\cite{Ehrgott:hard-to-say}). Thus, using the two given objective functions as the objectives~$f_i$ for two positions~$i$ with $\alpha_i=1$ (and arbitrary objective functions for all other positions) yields an instance with $p$~objectives where exponentially many solutions are required in any $(\alpha_1,\ldots,\alpha_p)$-approximate Pareto set.
\end{remark}



In contrast to this, Papadimitriou and Yannakakis~\cite{Papadimitriou+Yannakakis:multicrit-approx} show that if $\alpha_i > 1$ for all~$i$, there always exists an $(\alpha_1,\ldots,\alpha_p)$-approximate Pareto set of polynomial cardinality. In this paper, we focus on the case where $\alpha_i = 1$ for exactly one~$i$. Thus, we study approximate Pareto sets where, for any feasible solution~$x$, there exists a solution in the approximate Pareto set that has value no worse than~$f_i(x)$ in objective~$f_i$ and simultaneously achieves an approximation factor of~$1+\varepsilon$ in all other objective functions for some $\varepsilon>0$. For simplicity, we assume that the first objective~$f_1$ is to be optimized exactly, i.e., that $\alpha_1 = 1$ and $\alpha_j=1+\varepsilon$ for $j = 2,\ldots,p$. 


\begin{definition}[One-Exact $\varepsilon$-Pareto Set]
	For $\varepsilon > 0$, a $(1,1+\varepsilon,\ldots,1+\varepsilon)$-approximate Pareto set is called a \emph{one-exact $\varepsilon$-Pareto set}.
\end{definition}


A common way of dealing with multiobjective optimization problems are scalarizations, which turn the multiple objective functions into one objective function in some useful way. The resulting singleobjective optimization problem can be solved using known methods from singleobjective optimization and the obtained solution can then be used in the process of solving the multiobjective problem. One of the most common scalarization methods consists of putting some upper bound/budget on all objective functions but one, which is then minimized subject to the resulting budget constraints on the other objectives (see, e.g.,~\cite{Ehrgott:book}). 

\begin{definition}[Budget-Constrained Problem ({\normalfont$\constrained$})]
	Given an instance~$(X,f)$ of a multiobjective minimization problem and bounds $B_i>0$, $i = 2,\ldots,p$, for all objective functions except the first one, the subproblem $\constrained(B_2\ldots,B_p)$ is the following: Either answer that there does not exist a feasible solution $x' \in X$ with $$f_i(x') \leq B_i, \qquad i = 2,\ldots, p,$$ or return a feasible solution that minimizes $f_1$ among all such solutions, i.e., return $x \in X$ with \begin{align*}f_1(x) &=  \opt_1(B_2,\ldots,B_p) \colonequals  \min_{x'\in X} \{f_1(x'):\,f_i(x') \leq B_i \textnormal{ for } i = 2,\ldots,p\},\\
	f_i(x) &\leq  B_i, \qquad  i = 2,\ldots, p.
	\end{align*}
\end{definition}

Even though this scalarization via budget constraints is widely used both in practice and in the theoretical literature on multiobjective optimization, the $\constrained$ problem is hard to solve even for the biobjective case of many relevant problems such as $\shortestpath$. However, there often exists a PTAS, i.e., a polynomial-time algorithm that finds an arbitrarily good approximation. The problem of finding a $(1+\delta)$-approximation for $\constrained$ for any given value $\delta > 0$ is called the \emph{restricted problem}~\cite{Diakonikolas+Yannakakis:approx-pareto-sets}:

\begin{definition}[Restricted Problem ({\normalfont$\restrict$})]
		Given an instance~$(X,f)$ of a multiobjective minimization problem, bounds $B_i>0$, $i = 2,\ldots,p$, for all objective functions except the first one, and $\delta > 0$, the subproblem $\restrict_\delta(B_2,\ldots, B_p)$ is the following: Either answer that there does not exist a feasible solution $x' \in X$ with $$f_i(x') \leq B_i, \qquad i = 2,\ldots,p,$$ or return $x \in X$ with 
	\begin{align*}f_1(x) &\leq (1+\delta) \cdot \opt_1(B_2,\ldots,B_p),&\\
	f_i(x) &\leq B_i,&&   i = 2,\ldots, p.\end{align*}	
\end{definition}

An alternative way of circumventing the hardness of the budget-constrained problem is to consider solutions that violate the given bounds slightly, while requiring an objective value that is at least as good as the objective value of any solution that respects the bounds~\cite{Diakonikolas+Yannakakis:approx-pareto-sets}:

\begin{definition}[{\normalfont$\dualrestrict$}]
		Given an instance~$(X,f)$ of a multiobjective minimization problem, bounds $S_i>0$, $i = 2,\ldots,p$, for all objectives except the first one, and $\delta > 0$, the subproblem $\dualrestrict_\delta(S_2,\ldots, S_p)$ is the following: Either answer that there does not exist a feasible solution $x' \in X$ with $$f_i(x') \leq S_i, \qquad i = 2,\ldots,p,$$ or return $x \in X$ with 
	\begin{align*}f_1(x) &\leq \opt_1(S_2,\ldots,S_p),&\\
	f_i(x) &\leq (1+\delta) \cdot S_i,&&   i = 2,\ldots, p.\end{align*}	
\end{definition}

Note that, in an instance of $\dualrestrict$, the case might occur where there does not exist any feasible solution $x' \in X$ with $f_i(x') \leq S_i$ for $i = 2,\ldots, p$, but there exists a solution~$x \in X$ with $f_i(x) \leq (1+\delta) \cdot S_i$ for $i = 2,\ldots, p$. In this case, $\textnormal{NO}$ is a correct answer to the $\dualrestrict$ instance, but, since $\opt_1(S_2, \ldots, S_p) = +\infty > f_1(x)$, also~$x$ is a correct answer. Thus, for $\dualrestrict$, there are situations where both of the distinguished cases apply, whereas the two considered cases are always disjoint for $\constrained$ and $\restrict$. Also note that $\constrained$ can be viewed as the limit case $\delta = 0$ for both $\restrict$ and $\dualrestrict$. 

\smallskip

All three of the above subproblems can also be defined such that, instead of the first one, some other objective is to be optimized subject to budgets on the rest of the objectives. In the following, we sometimes use a superscript to indicate which objective is to be  optimized. For example, $\restrict^i$ denotes the restricted problem with a bound on all objectives but the $i$-th one.

\section{Existence and Cardinality of One-Exact $\varepsilon$-Pareto Sets}
Papadimitriou and Yannakakis~\cite{Papadimitriou+Yannakakis:multicrit-approx} show that, for any instance of a multiobjective optimization problem, there exists an $\varepsilon$-Pareto set whose cardinality is polynomial in the encoding length of the instance and in $\frac 1 \varepsilon$. Similarly, we now show the existence of one-exact $\varepsilon$-Pareto sets of polynomial cardinality. 

\begin{theorem}\label{thm:existence}
	For any $p$-objective optimization problem~$(X,f)$ and any given $\varepsilon > 0$, there exists a one-exact $\varepsilon$-Pareto set of cardinality $\mathcal{O}((\frac M \varepsilon)^{p-1})$.
\end{theorem}
\begin{proof}
	Figure~\ref{fig:existence} illustrates the proof.
	
	Consider the hypercube $[2^{-M},2^M] \times \ldots \times [2^{-M},2^M]$, in which all the feasible points are contained, and cover this hypercube by hyperstripes of the form
	$[2^{-M},2^M] \times [(1+\varepsilon)^{i_2},(1+\varepsilon)^{i_2+1}] \times \ldots \times [(1+\varepsilon)^{i_p},(1+\varepsilon)^{i_p+1}]$, for all~$i_2,\ldots,i_p \in \{-\lceil\frac{M}{\log (1+\varepsilon)}\rceil, \ldots, -1,0,1, \ldots, \lceil\frac{M}{\log (1+\varepsilon)}\rceil-1\}$. Note that, for this covering, we use $(2\cdot \lceil\frac{M}{\log (1+\varepsilon)}\rceil)^{p-1} = \mathcal{O}((\frac M \varepsilon)^{p-1})$ many hyperstripes.
	
	For each hyperstripe~$H$ containing feasible points, we choose one feasible point $y = f(x)\in H$ with minimum $f_1$-value among all feasible points in~$H$. Then all points in~$H$ are $(1,1+\varepsilon,\ldots,1+\varepsilon)$-dominated by~$y$. Thus, keeping one solution~$x \in X$ for each chosen point $y = f(x)$ (where points that are dominated by other chosen points can be discarded) yields a one-exact $\varepsilon$-Pareto set.
	Since at most one solution is chosen for each hyperstripe, the cardinality of the constructed set is in $\mathcal{O}((\frac M \varepsilon)^{p-1})$.	
	\qed
\end{proof}

\begin{figure}[ht!]
	\begin{center}
		\begin{tikzpicture}[scale=1]
		\draw[->] (-0.2,0) -- (7.5,0) node[below right] {$f_1$};
		\draw[->] (0,-0.2) -- (0,7.4) node[above left] {$f_2$};

		\foreach \x in {0.6, 0.9, 1.35, 2.03, 3.04, 4.56, 6.83}{
			\draw[-,line width=0.3pt] (0.6,\x) -- (6.83,\x);
			
		}
		\foreach \x in {0.6, 6.83}{
			\draw[-,line width=0.3pt] (\x,0.6) -- (\x,6.83);	
		}
		\foreach \x in {0.7, 6}{
			\draw[dashed, line width=0.3pt] (\x,0.7) -- (\x,6);
			\draw[dashed, line width=0.3pt] (0.7,\x) -- (6,\x);
		}
		\fill (0.84,5.90) circle (2pt);
		\fill (1.52,2.23) circle (2pt);
		\fill (1.72,1.63) circle (2pt);
		\draw (1.95,4.43) circle (2pt);
		\fill (3.24,2.63) circle (1pt);
		
		\fill (1.03,4.67) circle (1pt);
		\fill (1.44,5.63) circle (1pt);
		\fill (1.62,2.83) circle (1pt);
		\fill (1.83,5.46) circle (1pt);
		\fill (1.92,2.60) circle (1pt);
		\fill (1.93,1.77) circle (1pt);
		\fill (2.17,1.49) circle (1pt);
		\fill (2.23,3.46) circle (1pt);
		\fill (2.23,4.73) circle (1pt);
		\fill (2.53,2.15) circle (1pt);
		\fill (2.73,2.58) circle (1pt);
		\fill (2.83,4.06) circle (1pt);
		\fill (3.14,1.42) circle (1pt);
		\fill (3.27,0.99) circle (2pt);
		\fill (3.29,4.23) circle (1pt);
		\fill (3.34,3.37) circle (1pt);
		\fill (3.51,2.17) circle (1pt);
		\fill (3.60,1.13) circle (1pt);
		\fill (3.69,3.63) circle (1pt);
		\fill (3.90,3.85) circle (1pt);
		\fill (4.07,3.23) circle (1pt);
		\fill (4.29,2.79) circle (1pt);
		\fill (4.46,1.93) circle (1pt);
		
		\draw[-] (0.70,0.1) -- (0.70,-0.1) node[below] {$2^{-\!M}$};
		\draw[-] (6.0,0.1) -- (6.0,-0.1) node[below] {$2^M$};
		\draw[-] (0.1,0.7) -- (-0.1,0.7) node[left] {$2^{-M}$};
		\draw[-] (0.1,1.35) -- (-0.1,1.35) node[left] {$(1+\varepsilon)^{-1}$};
		\draw[-] (0.1,2.03) -- (-0.1,2.03) node[left] {$1$};
		\draw[-] (0.1,3.04) -- (-0.1,3.04) node[left] {$(1+\varepsilon)^{1}$};
		\draw[-] (0.1,6.0) -- (-0.1,6.0) node[left] {$2^M$};
		\draw[-] (0.1,6.83) -- (-0.1,6.83) node[left] {$(1\!+\!\varepsilon)^{\lceil\!\frac{M}{\log (1+\varepsilon)}\!\rceil}$};
		\end{tikzpicture}
		\caption{Proof of existence of polynomial-size one-exact $\varepsilon$-Pareto sets. Choose a feasible point minimizing~$f_1$ in each hyperstripe that contains at least one feasible point. One can discard points that are dominated by other chosen points. In the picture, feasible points are marked by dots, chosen points are indicated by thick dots, and discarded points are drawn as circles.\label{fig:existence}}
	\end{center}
\end{figure}
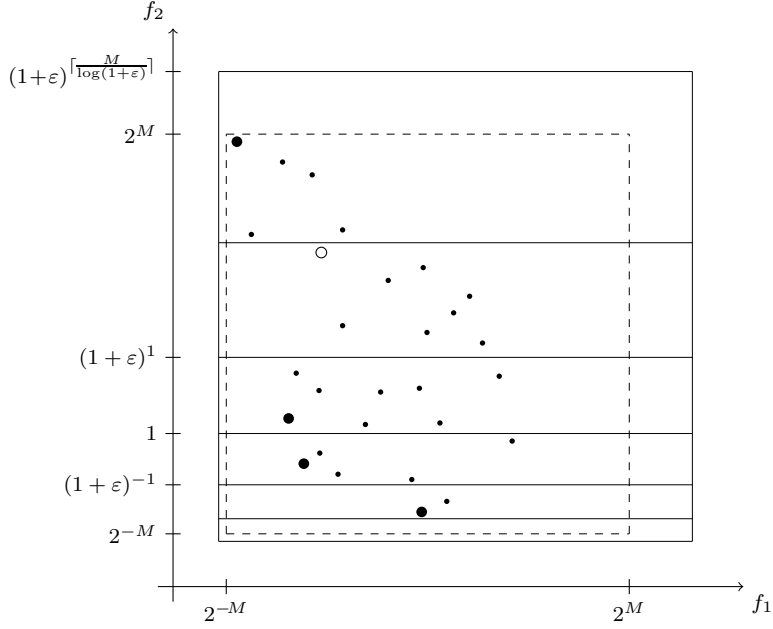

It is easy to see that $(1,1,1+\varepsilon)$-approximate Pareto sets of polynomial cardinality do not exist in general since this would imply the existence of polynomial (exact) Pareto sets for biobjective problems (see Remark~\ref{rmk:no-two-ones}). This means that, in this sense, an approximation factor of $(1,1+\varepsilon,\ldots,1+\varepsilon)$ is the best one achievable with polynomially many solutions.

In general, even a smallest $\varepsilon$-Pareto set may require $\Omega((\frac M \varepsilon)^{p-1})$ many solutions~\cite{Papadimitriou+Yannakakis:multicrit-approx}, which equals the worst-case bound on the cardinality of a smallest one-exact $\varepsilon$-Pareto set obtained from Theorem~\ref{thm:existence}. For a distinct instance, however, the two can be very different in size. Any one-exact $\varepsilon$-Pareto set is, in particular, an $\varepsilon$-Pareto set, so, for any instance, a one-exact $\varepsilon$-Pareto set of minimum cardinality is at least as large as an $\varepsilon$-Pareto set of minimum cardinality. In the other direction, the following holds:


\begin{theorem}\label{thm:unbounded}
	For any $\varepsilon>0$ and any positive integer~$n\in\mathbb{N}_+$, there exist instances of biobjective optimization problems such that $|P^*| > n\cdot |P^*_{\varepsilon}|$, where~$P^*$ denotes a smallest one-exact $\varepsilon$-Pareto set and $P^*_{\varepsilon}$ denotes a smallest $\varepsilon$-Pareto set.
\end{theorem}
\begin{proof}
	Given $\varepsilon > 0$ and $n \in \mathbb N_{+}$, we construct an instance of a biobjective minimization problem with $|P^*| = n+1$ and $|P^*_{\varepsilon}| = 1$.
	
	\noindent
	Let $X \colonequals \{x_0,\ldots, x_{n}\}$ and, for $i = 0,\ldots, n$, let $$f(x_i) = \left(f_1(x_i),f_2(x_i)\right) \colonequals \left(1+\frac{n-i} n \cdot \varepsilon, (1+\varepsilon)^{2i}\right).$$	
	Then, we have
	$$f_1(x_0) > f_1(x_1) > \ldots > f_1(x_n)$$ and
	$$f_2(x_0) < \frac 1 {1+\varepsilon} \cdot f_2(x_1) < \ldots < \frac 1 {(1+\varepsilon)^n} \cdot f_2(x_n),$$
	so no solution $(1,1+\varepsilon)$-approximates any other solution and, thus, $P^* = X$. However, we also have $$f_1(x_0) = 1+\frac {n-0} n \cdot \varepsilon = 1+\varepsilon = (1+\varepsilon) \cdot f_1(x_n),$$
	so $x_0$ $(1+\varepsilon,1+\varepsilon)$-approximates all other solutions and, thus,  $\{x_0\}$ is a one-exact $\varepsilon$-Pareto set. This construction is depicted in Figure~\ref{fig:unbounded}.
	\qed
\end{proof}

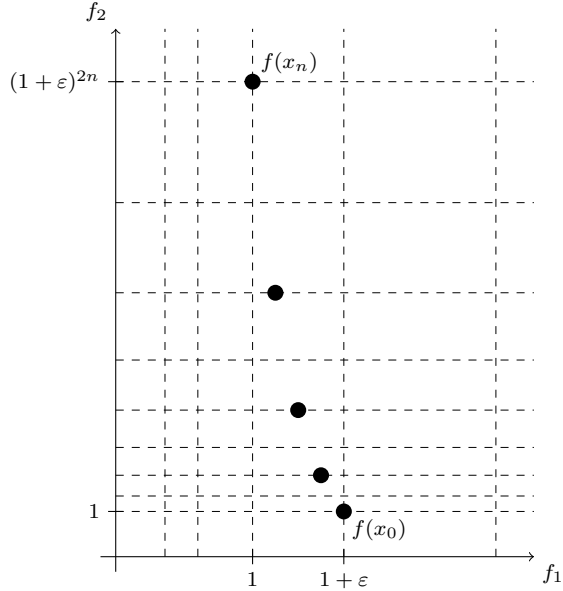
\begin{figure}[ht!]
	\begin{center}
		\begin{tikzpicture}[scale=1][decoration=brace]

		\fill (1.8,6.299) circle (3pt) node[above right] {$f(x_n)$};
		\fill (2.1,3.5) circle (3pt) node[above right] {};
		\fill (2.4,1.944) circle (3pt) node[above right] {};
		\fill (2.7,1.08) circle (3pt) node[above right] {};
		\fill (3,0.6) circle (3pt) node[below right] {$f(x_0)$};
		
		\foreach \x in {0.648, 1.08, 1.8, 3, 5}{
			\draw[dashed, line width=0.3pt] (\x,0) -- (\x,7);	
		}
		\foreach \x in {0.6, 0.805, 1.08, 1.449, 1.944, 2.608, 3.5, 4.695, 6.299}{
			\draw[dashed, line width=0.3pt] (0,\x) -- (5.5,\x);	
		}
		
		\draw[->] (-0.2,0) -- (5.5,0) node[below right] {$f_1$};
		\draw[->] (0,-0.2) -- (0,7) node[above left] {$f_2$}; 
		
		\draw[-] (1.8,0.1) -- (1.8,-0.1) node[below] {$1$};
		\draw[-] (3,0.1) -- (3,-0.1) node[below] {$1+\varepsilon$};
		%
		%
		\draw[-] (0.1,0.6) -- (-0.1,0.6) node[left] {$1$};
		\draw[-] (0.1,6.299) -- (-0.1,6.299) node[left] {$(1+\varepsilon)^{2n}$};
		
		\end{tikzpicture}
		\caption{Illustration of the instance in the proof of Theorem~\ref{thm:unbounded}. The solution~$x_0$ $(1+\varepsilon,1+\varepsilon)$-approximates all other solutions, but no solution $(1,1+\varepsilon)$-approximates any other solution.\label{fig:unbounded}}
	\end{center}
\end{figure}

Note that, in the instance constructed in the proof of Theorem~\ref{thm:unbounded}, we even have $|P^*| = \Omega(\frac M \varepsilon)$, i.e., a smallest one-exact $\varepsilon$-Pareto set, in fact, has the worst-case size, while a smallest one-exact $\varepsilon$-Pareto set~$P^*_{\varepsilon}$ consists of only one solution.

Moreover, the statement of Theorem~\ref{thm:unbounded} also holds for problems with three or more objectives. For any $p\geq 2$, one can similarly construct instances of $p$-objective problems where a smallest one-exact $\varepsilon$-Pareto set has the worst-case size of $\Omega((\frac M \varepsilon)^{p-1})$, while a smallest $\varepsilon$-Pareto set has cardinality one.

\section{Polynomial-Time Computability of One-Exact Pareto Sets}
The proof of Theorem~\ref{thm:existence} can easily be turned into a method for computing one-exact $\varepsilon$-Pareto sets that runs in fully polynomial time if and only if $\constrained^1$ is solvable in polynomial time. In the appendix, we present a method that, for biobjective problems, even computes a smallest one-exact $\varepsilon$-Pareto set using a subroutine for $\constrained$. However, $\constrained$ is \textsf{NP}-hard to solve for many relevant problems (a notable exception being multiobjective linear programming).

Instead, we now provide a method that computes a one-exact $\varepsilon$-Pareto set in (fully) polynomial time if a (fully) polynomial method for $\dualrestrict^1_\delta$ is available. This is the case for a significantly larger class of (relevant) problems including important problems such as multiobjective $\shortestpath$ and multiobjective $\spanningtree$. The method is based on the following lemma, which is visualized in Figure~\ref{fig:dualrestrict}.

\begin{lemma}\label{lem:duallemma}
	Let $x\in X$ be a solution to $\dualrestrict^1_\delta(S_2,\ldots,S_p)$, where $0 < \delta < \varepsilon$ for some $\varepsilon > 0$. Then any feasible point in the hyperstripe $$H = \left[2^{-M},2^M\right] \times \left[ \frac {1+\delta}{1+\varepsilon}S_2,S_2 \right] \times \ldots \times \left[ \frac {1+\delta}{1+\varepsilon}S_p,S_p \right] $$ is $(1,1+\varepsilon, \ldots,1+\varepsilon)$-dominated by f(x).
\end{lemma}
\begin{proof}
	By definition of $\dualrestrict$, we know that, in the first objective, $f_1(x) \leq \opt_1(S_2,\ldots,S_p)$, so there does not exist a feasible solution $x' \in X$ that satisfies $f_1(x') < f_1(x)$ and $f_i(x') \leq S_i$ for all $i = 2,\ldots,p$. We also know that $f_i(x) \leq (1+\delta) \cdot S_i$ for $i = 2,\ldots,p$.
	
	Not let $f(x') \in H$ be a feasible point in the hyperstripe. Then, since $f_i(x') \leq S_i$ for $i = 2,\ldots, p$, we must have $f_1(x') \geq f_1(x)$. Moreover, we have $f_i(x') \geq \frac {1+\delta}{1+\varepsilon} S_i$, which yields that
	$$(1+\varepsilon) \cdot f_i(x') \geq (1+\varepsilon) \cdot \frac {1+\delta} {1+\varepsilon} \cdot S_i = (1+\delta) \cdot S_i \geq f_i(x),$$
	so $f(x')$ is $(1,1+\varepsilon,\ldots,1+\varepsilon)$-dominated by~$f(x)$.\qed
\end{proof}	

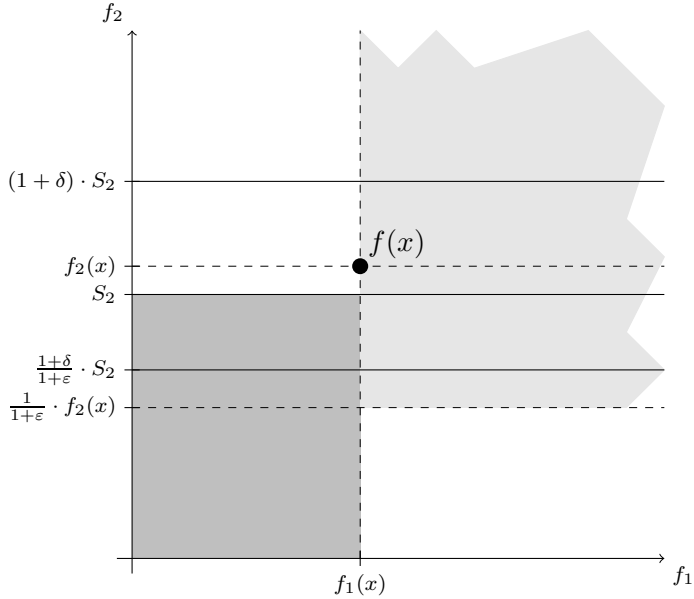
\begin{figure}[ht!]
	\begin{center}
		\begin{tikzpicture}[scale=1][decoration=brace]
		
		
		\fill[gray!50] (0,0) rectangle (3,3.5);
		\fill[gray!20] (3,2) rectangle (6.5,6.5);
		\draw[fill,gray!20] (6.5,6.5) -- (7,6) -- (6.5,4.5)--(7,4) -- (6.5,3) -- (7,2.5) -- (6.5,2);
		
		\draw[fill,gray!20] (6.5,6.5) -- (6,7) -- (4.5,6.5)--(4,7) -- (3.5,6.5) -- (3,7) -- (3,6.5);

		\fill (3,3.875) circle (3pt) node[above right] {\large{$f(x)$}};
		
		\foreach \x in {2.5,3.5,5}{
			\draw[-,thin] (0,\x) -- (7,\x) ;
		}	
		\draw[dashed, line width=0.3pt] (0,3.875) -- (7,3.875);
		\draw[dashed, line width=0.3pt] (0,2) -- (7,2);
		\draw[dashed, line width=0.3pt] (3,0) -- (3,7);
		

		\draw[->] (-0.2,0) -- (7,0) node[below right] {$f_1$};
		\draw[->] (0,-0.2) -- (0,7) node[above left] {$f_2$};
		
		\draw[-] (3,0.1) -- (3,-0.1) node[below] {$f_1(x)$};

		\draw[-] (0.1,2) -- (-0.1,2) node[left] {$\frac 1 {1+\varepsilon} \cdot f_2(x)$};
		\draw[-] (0.1,2.5) -- (-0.1,2.5) node[left] {$\frac {1+\delta}{1+\varepsilon} \cdot S_2$};	
		\draw[-] (0.1,3.5) -- (-0.1,3.5) node[left] {$S_2$};
		\draw[-] (0.1,3.875) -- (-0.1,3.875) node[left] {$f_2(x)$};
		\draw[-] (0.1,5) -- (-0.1,5) node[left] {$(1+\delta) \cdot S_2$};
		
		\end{tikzpicture}
		\caption{Illustration of Lemma~\ref{lem:duallemma}. The point~$f(x)$ is the image of a solution $x \in X$ to $\dualrestrict^1_\delta(S_2)$ with $0 < \delta < \varepsilon$. The dark gray area does not contain any feasible point. Every feasible point in the light gray area is $(1,1+\varepsilon)$-dominated by~$f(x)$. \label{fig:dualrestrict}}
	\end{center}
\end{figure}

Note that, if $\textnormal{NO}$ is a solution to $\dualrestrict^1_\delta(S_2,\ldots,S_p)$, then the hyperstripe~$H$ considered in Lemma~\ref{lem:duallemma} does not contain any feasible point. Thus, we know a priori that solving $\dualrestrict^1_\delta$ takes care of~$H$ in the sense that, if~$H$ contains feasible points, then $\dualrestrict^1_\delta(S_2,\ldots,S_p)$ is guaranteed to yield a feasible solution~$x \in X$ that $(1,1+\varepsilon,\ldots,1+\varepsilon)$-approximates every feasible solution~$x' \in X$ with $f(x') \in H$.

If~$\delta$ is chosen such that $\delta \in \Omega(\varepsilon)$, e.g., such that $(1+\delta)^2 = 1+\varepsilon$, then the hypercube $[2^{-M},2^M]^p$, in which all feasible points are contained, can be covered by $\mathcal{O}((\frac M \varepsilon)^{p-1})$ many such hyperstripes, each of which, in turn, can be taken care of by one solution of $\dualrestrict^1_{\delta}$.

This idea of covering the range of possible objective values by polynomially many solutions of $\dualrestrict$ is used by Angel et al.~\cite{Angel+etal:min-cost-makespan,Angel+etal:bicriteria-scheduling} for the biobjective min-cost-makespan scheduling problem. We formalize the idea for general multiobjective optimization problems in Algorithm~\ref{alg:basicAlgo}.

\begin{algorithm}[!ht]
	\SetKw{Compute}{compute}
	\SetKw{Break}{break}
	\SetKwInOut{Input}{input}\SetKwInOut{Output}{output}
	\SetKwComment{command}{right mark}{left mark}
	
	\Input{an instance $(X,f)$ of a $p$-objective minimization problem, $\varepsilon>0$, an algorithm for $\dualrestrict^1_{\delta}$, where $\delta = \sqrt{1+\varepsilon} -1$}
	
	\Output{a one-exact $\varepsilon$-Pareto set for $(X,f)$}
	
	\BlankLine
	
	$P \leftarrow \emptyset$
	
	$\delta \leftarrow \sqrt{1+\varepsilon} -1$
	
	$u \leftarrow \lceil\frac{M}{\log (1+\delta)}\rceil$

	\ForEach {$(i_2,\ldots,i_p)$ such that  $i_{\ell} \in \{-u+1,\ldots,u\}$, $l = 1,\ldots,p$}{
		
		$x \leftarrow \dualrestrict^1_\delta((1+\delta)^{i_2},\ldots,(1+\delta)^{i_p})$
		
		\If{$x \neq \textnormal{NO}$}{
			
			
				
					
			%
			%
			
				
				$P \leftarrow P \cup \{x\}$
		}
	}
	\Return $P$
	\caption{A $(1,1+\varepsilon,\ldots,1+\varepsilon)$-approximation for multiobjective optimization problems}\label{alg:basicAlgo}
\end{algorithm}

\begin{theorem}
	For a given instance $(X,f)$ of a $p$-objective minimization problem and a given $\varepsilon > 0$, Algorithm~\ref{alg:basicAlgo} computes a one-exact $\varepsilon$-Pareto set. The algorithm solves $\mathcal{O}((\frac M \varepsilon)^{p-1})$ instances of $\dualrestrict^1_\delta$, where $(1+\delta)^2 = 1+\varepsilon$. The returned set~$P$ has (polynomial) cardinality $|P| = \mathcal{O}((\frac M \varepsilon)^{p-1})$.
\end{theorem}

\begin{proof}
	Algorithm~\ref{alg:basicAlgo} (implicitly) covers the hypercube $[2^{-M},2^M]\times ... \times [2^{-M},2^M]$ in the objective space by $(2u)^{p-1} = (2 \cdot \lceil\frac{M}{\log (1+\delta)}\rceil)^{p-1} = \mathcal{O}((\frac M \varepsilon)^{p-1})$ hyperstripes of the form
	\begin{align*}
	H &= \left[2^{-M},2^M\right] \times \left[(1+\delta)^{i_2-1},(1+\delta)^{i_2}\right] \times \ldots \times \left[(1+\delta)^{i_p-1},(1+\delta)^{i_p}\right]\\		
	&=\left[2^{-M},2^M\right] \times \left[ \frac {1\!+\!\delta} {1\!+\!\varepsilon}(1+\delta)^{i_2},(1+\delta)^{i_2} \right] \times \!\ldots\! \times \left[ \frac {1\!+\!\delta} {1\!+\!\varepsilon}(1+\delta)^{i_p},(1+\delta)^{i_p} \right]\!.
	\end{align*}
	Lemma~\ref{lem:duallemma} implies that, for each of these hyperstripes, solving the subproblem $\dualrestrict^1_\delta((1+\delta)^{i_2},\ldots,(1+\delta)^{i_p})$ either yields a feasible solution~$x \in X$ such that all feasible points in the hyperstripe are $(1,1+\varepsilon,\ldots,1+\varepsilon)$-dominated by~$f(x)$, or it yields $\textnormal{NO}$, which guarantees that the hyperstripe does not contain any feasible point.
	Hence, the set of all feasible solutions produced by solving $\dualrestrict^1_\delta$ for these hyperstripes is a one-exact $\varepsilon$-Pareto set of cardinality $\mathcal{O}((\frac M \varepsilon)^{p-1})$.
	\qed
\end{proof}

We note that the one-exact $\varepsilon$-Pareto set returned by Algorithm~\ref{alg:basicAlgo} may contain solutions that are dominated by other solutions in the set. Such solutions can be removed without influencing the obtained approximation quality. However, filtering out dominated solutions might actually require more time than computing the set itself in situations where $\dualrestrict$ can be solved very efficiently.


\smallskip

Papadimitriou and Yannakakis~\cite{Papadimitriou+Yannakakis:multicrit-approx} show that there is an equivalence between solving the $\gap$ problem associated with a multiobjective optimization problem and finding an $\varepsilon$-Pareto set in the sense that one can compute an $\varepsilon$-Pareto set in (fully) polynomial time if and only if one can solve $\gap$ in (fully) polynomial time.

We now prove an analogous result for $\dualrestrict$ and one-exact $\varepsilon$-Pareto sets. This demonstrates that $\dualrestrict$ is, in fact, exactly the right auxiliary problem to consider for computing one-exact $\varepsilon$-Pareto sets.

\begin{theorem}
	A one-exact $\varepsilon$-Pareto set for an instance~$I$ of a multiobjective optimization problem can be found for any~$\varepsilon > 0$ in time polynomial in the encoding length of~$I$ (and in $\frac 1 \varepsilon$) if and only if $\dualrestrict^1_\delta$ can be solved for any $\delta > 0$ in time polynomial in the encoding length of $I$ (and in $\frac 1 \delta$).
\end{theorem}

\begin{proof}
	If $\dualrestrict^1_\delta$ can be solved in (fully) polynomial time, a one-exact $\varepsilon$-Pareto set can be found in (fully) polynomial time using Algorithm~\ref{alg:basicAlgo}.
	
	\smallskip
	
	Conversely, suppose that we can compute a one-exact $\varepsilon$-Pareto set in (fully) polynomial time. Then, given bounds $S_2,\ldots,S_p > 0$, we can solve the problem  $\dualrestrict^1_\delta(S_2,\ldots,S_p)$ as follows:
	
	We start by computing a one-exact $\delta$-Pareto set~$P$. Then, if there is no solution~$x \in P$ with $f_i(x) \leq (1+\delta) \cdot S_i$ for $i = 2,\ldots, p$, we return $\textnormal{NO}$. This is a correct answer since, if there was a solution~$x'$ with $f_i(x) \leq S_i$ for $i = 2,\ldots, p$, there would be no solution~$x \in P$ that $(1,1+\delta,\ldots,1+\delta)$-approximates~$x'$ in contradiction to~$P$ being a one-exact $\delta$~Pareto set. If there exist solutions~$x \in P$ with $f_i(x) \leq (1+\delta) \cdot S_i$ for $i = 2,\ldots, p$, we return one of them with minimum value in~$f_1$. Assume that, for the returned solution~$x$, we have $f_1(x) > \opt_1(S_2,\ldots,S_p)$. Then this means that there is some $x' \in X$ with $f_i(x') \leq S_i$ for $i = 1,\ldots,p$ and $f_1(x') < f_1(x) = \min\{f_1(x''):x'' \in P,\; f_i(x'') \leq (1+\delta) \cdot S_i, i = 2,\ldots ,p\}$. Thus, $x'$ is not $(1,1+\delta,\ldots,1+\delta)$-approximated by any solution in~$P$, which again contradicts~$P$ being a one-exact $\delta$-Pareto set.
	\qed
\end{proof}

\section{Computing Small One-Exact $\varepsilon$-Pareto Sets}
In this section, we consider the question if and how we can compute one-exact $\varepsilon$-Pareto sets that are not only of polynomial size, but also guarantee some bound on the cardinality compared to the cardinality of a smallest one-exact $\varepsilon$-Pareto set~$P^*$.

The worst-case cardinality of a one-exact $\varepsilon$-Pareto set computed by Algorithm~1 is $(2 \cdot \lceil\frac{M}{\log (1+\delta)}\rceil)^{p-1}$ for $(1+\delta)^2 = 1+ \varepsilon$, which is a factor of~$2^{p-1}$ larger than the upper bound of $(2 \cdot\lceil\frac{M}{\log (1+\varepsilon)}\rceil)^{p-1}$ for $\varepsilon$-Pareto sets constructed in the proof of Theorem~\ref{thm:existence}. However, even when adding a filtering step that removes solutions dominated by other solutions in the computed set, Algorithm~\ref{alg:basicAlgo} does not provide an upper bound on the ratio $\frac {|P|}{|P^*|}$ for any fixed instance, where $P$ is a one-exact $\varepsilon$-Pareto set computed by Algorithm~\ref{alg:basicAlgo} and $P^*$ is a one-exact $\varepsilon$-Pareto set of minimum cardinality. 

With such an additional filtering step, it is possible to show that, for biobjective problems, we have an upper bound of~$4$ on the ratio~$\frac{|P|}{|P^*|}$. When using $\delta = \sqrt[3]{1+\varepsilon} -1$ instead of $\delta = \sqrt{1+\varepsilon} -1$ in Algorithm~\ref{alg:basicAlgo} and replacing $1+\delta$ by $(1+\delta)^2$ in lines~3 and~5, we can improve this to a ratio of~$3$. We can even achieve a ratio of~$2$ when setting $\delta = \sqrt[4]{1+\varepsilon} -1$ and using a more sophisticated elimination technique than simply filtering out dominated solutions.

Here, however, we derive a different algorithm for biobjective problems that does not operate on a predefined grid but instead uses adaptive steps in order to decrease the number of solved instances of $\dualrestrict$ while still ensuring a size guarantee of $\frac {|P|}{|P^*|} \leq 2$ even without an additional (potentially time-consuming) filtering step.

We first give some results that substantiate the hardness of computing one-exact $\varepsilon$-Pareto sets that are smaller than twice the minimum size. Then, we formulate our algorithm and prove its correctness. We additionally consider the cases that a efficient routine for solving $\restrict$ or $\constrained$ is given. Finally, we prove a result that indicates the hardness of achieving similar results for more than two objectives.

\subsection{Lower Bounds for Biobjective Problems}
The following result shows that, for biobjective optimization problems, any generic algorithm based on solving $\dualrestrict^1$ that computes a one-exact $\varepsilon$-Pareto set~$P$ of cardinality $|P| < 2 \cdot |P^*|$ needs to solve an instance of $\dualrestrict^1_\delta$ for some $\delta > 0$ for which $\frac 1 \delta$ is exponential in the encoding length of the input. Since the running time of a method for solving $\dualrestrict^1_{\delta}$ is typically at least linear in $\frac 1 \delta$, this implies that it is unlikely for such an algorithm to run in polynomial time. Note that, for problems where the running time of a routine for $\dualrestrict^1_\delta$ is at most logarithmic in $\frac 1 \delta$, we can solve $\constrained^1$ efficiently by setting $\delta < 2^{-2M}$ in $\dualrestrict^1_\delta$. Thus, in this case, we can even compute a smallest one-exact $\varepsilon$-Pareto set in polynomial time (see Corollary~\ref{cor:constrained}).

\begin{theorem} \label{thm:generictwo}
	For any $\varepsilon>0$, there does not exist an algorithm that computes a one-exact $\varepsilon$-Pareto set~$P$ such that $|P| < 2 \cdot |P^*|$ for every biobjective optimization problem and generates feasible solutions only via solving $\dualrestrict^1_\delta$ for values of~$\delta$ such that $\frac 1 \delta$ is polynomial in the encoding length of the input.
\end{theorem}
\begin{proof}
	Given $\varepsilon>0$, consider the two instances $I_1=(\{x_1,x_2\},f)$ and $I_2=(\{x_1,x_2,x_3\},f)$ in which $f(x_1) = (f_1(x_2)-1,(1+\varepsilon)\cdot f_2(x_2))$ and $f(x_3) = (f_1(x_2),f_2(x_2)-1)$ (the values~$f_1(x_2)$ and~$f_2(x_2)$ are defined later). Then~$\{x_1\}$ is a one-exact $\varepsilon$-Pareto set for~$I_1$. On the other hand, any one-exact $\varepsilon$-Pareto set for~$I_2$ needs at least two solutions since neither~$x_2$ nor~$x_3$ $(1,1+\varepsilon)$-approximates~$x_1$, and~$x_1$ does not $(1,1+\varepsilon)$-approximate~$x_3$.
	An algorithm that computes a one-exact $\varepsilon$-Pareto set~$P$ with $|P|<2 \cdot |P^*|$ would, therefore, have to be able to distinguish between~$I_1$ and~$I_2$, i.e., detect the existence of~$x_3$.
	
	Note that, for $S_2 < f_2(x_3)$, $\textnormal{NO}$ is a solution to $\dualrestrict^1_\delta(S_2)$ in both instances~$I_1$ and~$I_2$ for any~$\delta$. If $f_2(x_1) > S_2 \geq f_2(x_3)$ and $\delta \geq \frac 1  {f_2(x_3)}$, we have
	\begin{align*}
	f_2(x_2) = f_2(x_3) + 1 = (1 + \frac 1 {f_2(x_3)}) \cdot f_2(x_3) \leq (1+\delta) \cdot f_2(x_3) \leq (1+\delta) \cdot S_2,
	\end{align*}
	so~$x_2$ is a solution to $\dualrestrict_\delta^1(S_2)$ in both instances. For $S_2 \geq f_2(x_1)$, $x_1$ is a solution to $\dualrestrict^1_\delta(S_2)$ in both instances for any~$\delta$. Therefore, in order to tell the difference between~$I_1$ and~$I_2$, an algorithm using only $\dualrestrict^1$ to generate feasible solutions would have to solve $\dualrestrict_\delta^1(S_2)$ for~$S_2$ and~$\delta$ with $f_2(x_1) > S_2 \geq f_2(x_3)$ and $\delta < \frac 1  {f_2(x_3)}$, i.e., $\frac 1 \delta > f_2(x_3) = f_2(x_2) - 1$. Since the value~$f_2(x_2)$ might be exponentially large in the encoding length of the input, this proves the claim.
	\qed
\end{proof}

While Theorem~\ref{thm:generictwo} shows that generic algorithms based on $\dualrestrict$ cannot obtain a factor smaller than~$2$ with respect to the cardinality of a one-exact Pareto set without using exponentially large values of~$\frac{1}{\delta}$ (even if $\textsf{P}=\textsf{NP}$), we now show that, for certain problems, \emph{no} algorithm (whether based only on $\dualrestrict$ or not) can obtain a factor smaller than~$2$ under the assumption that $\textsf{P}\neq\textsf{NP}$. To this end, we consider the following biobjective scheduling problem: We are given a set~$J$ of $|J| = n$ independent jobs, which are to be scheduled on $m$~parallel machines. Performing job~$j$ on machine~$i$ takes processing time~$p_{ij}\geq 0$ and causes cost~$c_{ij}\geq 0$. The goal is to minimize the makespan (i.e., the maximum completion time of a job) and the total cost (i.e., the sum of the costs resulting from assigning jobs to machines). We call this problem, where the cost-objective is the first objective~$f_1$ and the makespan-objective is the second objective~$f_2$, the \emph{min-cost-makespan scheduling problem}.

The min-cost-makespan scheduling problem is known to have a fully po\-ly\-no\-mial-time algorithm for $\dualrestrict^1$ and a fully polynomial-time one-exact approximation algorithm due to Angel et al.~\cite{Angel+etal:min-cost-makespan,Angel+etal:bicriteria-scheduling}. For this problem, however, we can show the following additional hardness result regarding the computation of small one-exact $\varepsilon$-Pareto sets.

\begin{theorem}
	For the min-cost-makespan scheduling problem, if $0< \varepsilon< \frac 1 2$, it is \textsf{NP}-hard to compute a one-exact $\varepsilon$-Pareto set~$P$ of cardinality $|P| < 2 \cdot |P^*|$.
\end{theorem}

\begin{proof}
	We use a reduction from $\partition$. Given an instance $a_1,\ldots, a_n \in \mathbb N$ of $\partition$, where, without loss of generality, $A \colonequals \sum_{i=1}^n a_i \geq \frac 4 {1-2\varepsilon}$, define an instance of the min-cost-makespan scheduling problem as follows: We have $m = 2$ machines and $|J| = n+2$ jobs. For $j = 1,\ldots, n$, we have a job~$j$ with processing times $p_{1j} = p_{2j} = a_j$ and costs $c_{1j} = a_j$ and $c_{2j} = 0$. We have two additional jobs~$n+1$ and $n+2$, with $p_{1(n+1)} = p_{1(n+2)} = p_{2(n+1)} = p_{2(n+2)} = K$, $c_{1(n+1)} = c_{2(n+2)} = 1$, and $c_{2(n+1)} = c_{1(n+2)} = 2$, where $K>0$ is chosen such that
	\begin{align}
	\frac 1 {1+\varepsilon} \cdot(K+A) &> K+\frac A 2, \label{eq:np-cond1}\\ 
	\frac 1 {1+\varepsilon} \cdot(K+A) &\leq K+\frac A 2 + 1, \label{eq:np-cond2}\\
	\frac 1 {1+\varepsilon} \cdot(K+A) &\leq 2\cdot K.  \label{eq:np-cond3}
	\end{align}
	Note that it is possible to choose~$K$ like this by our assumptions on~$\varepsilon$ and~$A$. For instance, one can check that $K \colonequals \lceil \frac {1-\varepsilon}{2\varepsilon} \cdot A - 1 - \frac 1 \varepsilon\rceil$ fulfills \eqref{eq:np-cond1}--\eqref{eq:np-cond3} and $K > 0$.
	
	The schedule~$\bar s$ where machine~$1$ performs only~$\{n+1\}$ and machine~$2$ performs the jobs~$\{1,\ldots,n,n+2\}$ has a cost of $f_1(\bar s) = 2$. This is the unique minimum in~$f_1$ over all schedules, so the schedule~$\bar s$ is not $(1,1+\varepsilon)$-approximated by any other schedule. Thus, it must be part of every one-exact $\varepsilon$-Pareto set. Moreover, $\bar s$ has a makespan of $f_2(\bar s) = K+A$, so, by Inequality~\eqref{eq:np-cond3}, it $(1,1+\varepsilon)$-approximates every schedule where jobs~$n+1$ and~$n+2$ are performed on the same machine. 
	
	\smallskip
	
	If the instance of $\partition$ is a NO-instance, any schedule where jobs~$n+1$ and $n+2$ are performed on different machines has a makespan of at least $K+\frac A 2 + 1$, so, by Inequality~\eqref{eq:np-cond2}, the one-element set~$\{\bar s\}$ is a one-exact $\varepsilon$-Pareto set.
	
	\smallskip
	
	If the instance of $\partition$ is a YES-instance, i.e., if there exists a partition~$(I_1, I_2)$ such that $\sum_{i \in I_1} a_i = \sum_{i \in I_2} a_i = \frac A 2$, then the schedule where machine~$1$ performs jobs~$\{n+1\} \cup I_1$ and machine~$2$ performs jobs~$\{n+2\} \cup I_2$ has a makespan of $K+ \frac A 2$. By Inequality~\eqref{eq:np-cond1}, this schedule is not $(1,1+\varepsilon)$-approximated by~$\bar s$, so any one-exact $\varepsilon$-Pareto set must contain at least two solutions. Therefore, it is \textsf{NP}-hard to distinguish between the two cases $|P^*| = 1$ and $|P^*| \geq 2$. \qed
\end{proof}

\subsection{Algorithm for Biobjective Problems}

We now provide an algorithm that computes a one-exact $\varepsilon$-Pareto set~$P$ that is not larger than twice the cardinality of a smallest one-exact $\varepsilon$-Pareto set~$P^*$. The algorithm is formally stated in Algorithm~\ref{alg:twoapproxAlgo} and an illustration of its behavior is given in Figure~\ref{fig:twoapproxAlgo}. The following lemma collects several invariants that hold during the execution of the algorithm.























\begin{algorithm}[!ht]
	\SetKw{Compute}{compute}
	\SetKw{Break}{break}
	\SetKw{Breakrepeat}{break(repeat)}
	\SetKw{Breakwhilerepeat}{break(while,repeat)}
	\SetKwInOut{Input}{input}\SetKwInOut{Output}{output}
	\SetKwComment{command}{right mark}{left mark}
	
	\Input{an instance $(X,f)$ of a biobjective minimization problem, $\varepsilon>0$, an algorithm for $\dualrestrict^1$}
	
	\Output{a one-exact $\varepsilon$-Pareto set for $(X,f)$}
	
	\BlankLine
	
	$P \leftarrow \emptyset$
	
	$\delta \leftarrow \sqrt[4]{1+\varepsilon} -1$
	
	$S_2 \leftarrow 2^M$
	
	\Repeat{
		
		$x \leftarrow \dualrestrict^1_\delta(S_2)$
		
		\lIf{$x = \textnormal{NO}$}{\Breakrepeat}
		
		$S_2 \leftarrow \frac {1}{(1+\delta)^2}\cdot f_2(x)$
		
		$x^{\textnormal{next}} \leftarrow \dualrestrict^1_\delta(S_2)$
		
		\lIf{$\xnext = \textnormal{NO}$}{$P \leftarrow P \cup \{x\}$ and \Breakrepeat}
		
		\While {$f_1(\xnext) = f_1(x)$}{
			
			$x \leftarrow x^{\textnormal{next}}$
			
			$S_2 \leftarrow \frac {1}{(1+\delta)^2}\cdot f_2(x)$
			
			$x^{\textnormal{next}} \leftarrow \dualrestrict^1_\delta(S_2)$
			
			\lIf{$\xnext = \textnormal{NO}$}{$P \leftarrow P \cup \{x\}$ and \Breakwhilerepeat}
		}
		$P \leftarrow P \cup \{x\}$
		
		$S_2 \leftarrow \frac {1}{1+\varepsilon}\cdot f_2(x)$		
	}
	
	\Return $P$
	
	\caption{A $(1,1+\varepsilon)$-approximation for biobjective optimization problems} \label{alg:twoapproxAlgo}
\end{algorithm}

\begin{figure}[ht!]
	\begin{center}
		\begin{tikzpicture}[scale=1]
		
		\draw[fill,gray!20] (6.5,9.6) -- (6.7,9.6) -- (7,9) -- (6.7,8) -- (7,7) -- (6.7,6.5) -- (7,6) -- (6.7,4.5)--(7,4) -- (6.7,3) -- (7,2.5) -- (6.7,2) -- (7,1.5) -- (6.7,0.8) -- (6.5,0.8);		
		\draw[fill,gray!50] (6.7,9.6) -- (6,10.6) -- (4.5,10.3)--(4,10.6) -- (3.5,10.3) -- (3,10.6) -- (2,10.3) -- (1,10.6) -- (0,10.3) -- (0,9.6);
		\draw[fill,gray!50] (6.5,0) -- (6.7,0) -- (7,0.3) -- (6.7, 0.9) -- (6.5,0.9);
		
		
		\fill[gray!20] (1,0) rectangle (6.5,9.6);
		
		\fill[gray!50] (0,0) rectangle (1,9.6);
		\fill[gray!50] (1,0) rectangle (1.5,6.0);
		\fill[gray!50] (1.5,0) rectangle (2.5,4.8);
		\fill[gray!50] (2.5,0) rectangle (3,3.9);
		\fill[gray!50] (3,0) rectangle (4,2.7);
		\fill[gray!50] (4,0) rectangle (6.5,0.9);

		\draw (1.0,9.0) circle (3pt) node[above right] {$f( x^{(1)})$};
		\draw (1.0,8.1) circle (3pt) node[above right] {$f( x^{(2)})$};
		\fill (1.0,7.2) circle (3pt) node[above right] {$f( x^{(3)})$};
		\draw (1.5,6.3) circle (3pt) node[above right] {$f( x^{(4)})$};	
		\fill (2.5,5.1) circle (3pt) node[above right] {$f( x^{(5)})$};
		\draw (3.0,4.2) circle (3pt) node[above right] {$f( x^{(6)})$};
		\draw (4.0,3.0) circle (3pt) node[above right] {$f( x^{(7)})$};
		\fill (4.0,2.1) circle (3pt) node[above right] {$f( x^{(8)})$};
		
		\foreach \x in {9.6, 7.8, 6.9, 6.0, 4.8, 3.9, 2.7, 1.8, 0.9}{
			\draw[-,thin] (0,\x) -- (7,\x) ;
		}
		
		\draw[dashed, line width=0.3pt] (1,0) -- (1,9.6);
		\draw[dashed, line width=0.3pt] (1.5,0) -- (1.5,6);
		\draw[dashed, line width=0.3pt] (2.5,0) -- (2.5,4.8);	
		\draw[dashed, line width=0.3pt] (3,0) -- (3,3.9);	
		\draw[dashed, line width=0.3pt] (4,0) -- (4,2.7);	
		
		\draw[->] (-0.2,0) -- (7,0) node[below right] {$f_1$};
		\draw[->] (0,-0.2) -- (0,10.6) node[above left] {$f_2$}; 
		
		\draw[-] (1,0.1) -- (1,-0.1) node[below] {$f_1(x_1)$};
		\draw[-] (2.5,0.1) -- (2.5,-0.1) node[below] {$f_1(x_2)$};
		\draw[-] (4,0.1) -- (4,-0.1) node[below] {$f_1(x_3)$};

		\draw[-] (0.1,9.6) -- (-0.1,9.6) node[left] {$2^M \equalscolon S_2^{(1)}$};
		\draw[-] (0.1,7.8) -- (-0.1,7.8) node[left] {$\frac{1}{(1+\delta)^2} f_2( x^{(1)}) \equalscolon S_2^{(2)}$};
		\draw[-] (0.1,6.9) -- (-0.1,6.9) node[left] {$\frac{1}{(1+\delta)^2} f_2( x^{(2)}) \equalscolon S_2^{(3)}$};
		\draw[-] (0.1,6.0) -- (-0.1,6.0) node[left] {$\frac{1}{(1+\delta)^2} f_2( x^{(3)}) \equalscolon S_2^{(4)}$};
		\draw[-] (0.1,4.8) -- (-0.1,4.8) node[left] {$\frac{1}{1+\varepsilon} f_2( x^{(3)}) \equalscolon S_2^{(5)}$};
		\draw[-] (0.1,3.9) -- (-0.1,3.9) node[left] {$\frac{1}{(1+\delta)^2} f_2( x^{(5)}) \equalscolon S_2^{(6)}$};
		\draw[-] (0.1,2.7) -- (-0.1,2.7) node[left] {$\frac{1}{1+\varepsilon} f_2( x^{(5)}) \equalscolon S_2^{(7)}$};
		\draw[-] (0.1,1.8) -- (-0.1,1.8) node[left] {$\frac{1}{(1+\delta)^2} f_2( x^{(7)}) \equalscolon S_2^{(8)}$};
		\draw[-] (0.1,0.9) -- (-0.1,0.9) node[left] {$\frac{1}{(1+\delta)^2} f_2( x^{(8)}) \equalscolon S_2^{(9)}$};
		\draw[-] (0.1,0.3) -- (-0.1,0.3) node[left] {$2^{-M}$};
		
		\end{tikzpicture}
		\caption{An illustration of Algorithm~\ref{alg:twoapproxAlgo} in the objective space (on a logarithmic scale). Each~$ x^{(i)}$ is a solution to $\dualrestrict^1_\delta( S_2^{(i)})$ for $i = 1,\ldots,8$, where $(1+\delta)^4 = (1+ \varepsilon)$. The dark gray area does not contain any feasible point. Any feasible point in the light gray area is $(1,1+\varepsilon)$-dominated by $f( x^{(i)})$ for some~$i \in \{1,\ldots,8\}$. The solution~$x^{(4)}$ is discarded since any solution that is $(1,1+\varepsilon)$-approximated by~$x^{(4)}$ is also $(1,1+\varepsilon)$-approximated by $x^{(3)}$ or $x^{(5)}$.  The solution~$x^{(6)}$ is discarded since any solution that is $(1,1+\varepsilon)$-approximated by~$x^{(6)}$ is also $(1,1+\varepsilon)$-approximated by~$x^{(5)}$ or~$x^{(7)}$. The solutions~$x^{(1)}$, $ x^{(2)}$, and $x^{(7)}$ are discarded since they are dominated by $x^{(3)}$, $ x^{(3)}$, and $x^{(8)}$, respectively. $\dualrestrict_\delta(S_2^{(9)})$ returns NO, so the algorithm returns $\{x^{(3)}, x^{(5)}, x^{(8)}\}$. \label{fig:twoapproxAlgo}}
	\end{center}
\end{figure}
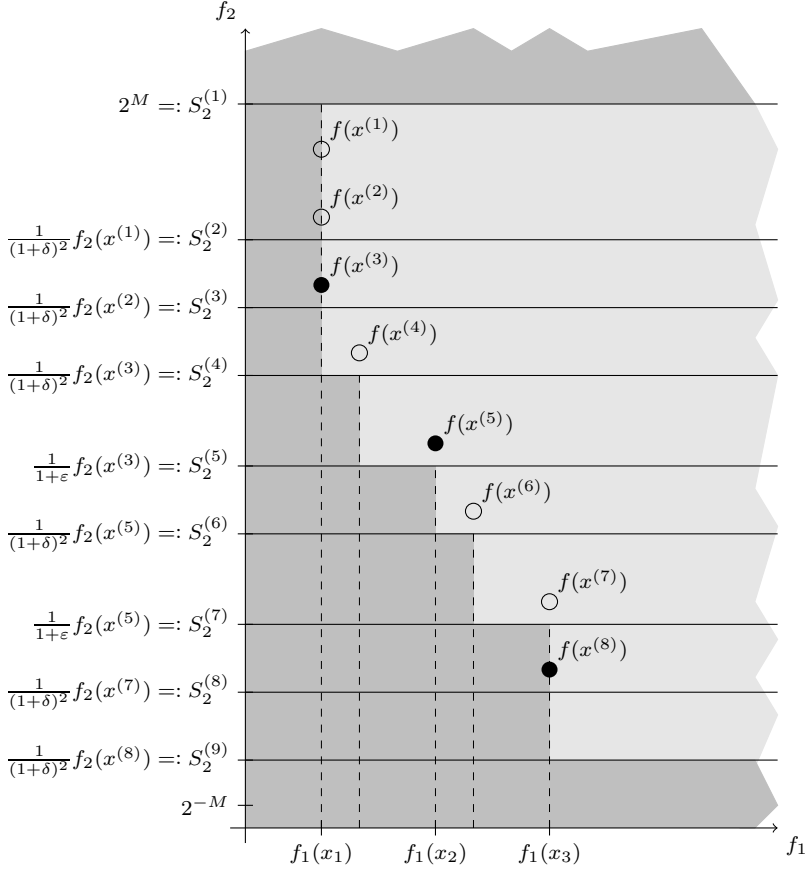

\begin{lemma} \label{lem:twoapproxlemma}
	When performing line~10 in any iteration of the while/repeat loops in Algorithm~\ref{alg:twoapproxAlgo}, the following properties hold:
	\begin{enumerate}[(a)]
		\item $x \neq \textnormal{NO}$ and $\xnext \neq \textnormal{NO}$,
		\item $S_2 = \frac 1 {(1+\delta)^2} \cdot f_2(x)$,
		\item $f_1(\xnext) \leq \opt_1(S_2)$,
		\item $f_2(\xnext) \leq \frac 1 {1+\delta} \cdot f_2(x)$,
		\item $f_1(\xnext) \geq f_1(x)$,
		\item The solutions in $P\cup \{x\}$ do $(1,1+\varepsilon)$-approximate all solutions~$x' \in X$ with $f_2(x') \geq \frac 1 {1+\varepsilon} \cdot f_2(x)$,
		\item $f_2(x) \leq \frac 1 {(1+\delta)^3} \cdot \min \{f_2(\bar x) : \bar x \in P\}$.
	\end{enumerate}
\end{lemma}
\begin{proof}
	\begin{enumerate} [(a)]
		\item As soon as solving an instance of $\dualrestrict_\delta^1$ yields $\textnormal{NO}$ during the execution of Algorithm~\ref{alg:twoapproxAlgo}, the repeat loop breaks immediately, so line~10 is not reached anymore.
		\item Before line~10, either line~7 or line~12 is executed.
		\item Before line~10, either line~8 or line~13 is executed. Therefore we have $\xnext = \dualrestrict_\delta^1(S_2)$, which implies $f_1(\xnext) \leq \opt_1(S_2)$.
		\item  Again, we have  $\xnext = \dualrestrict_\delta^1(S_2)$. This implies that $f_2(\xnext) \leq (1+\delta) \cdot S_2$. Using~(b), we obtain that $f_2(\xnext) \leq (1+\delta) \cdot \frac 1 {(1+\delta)^2} \cdot f_2(x) = \frac 1 {1+\delta} \cdot f_2(x)$.
		\item Since the only way feasible solutions are generated in Algorithm~\ref{alg:twoapproxAlgo} is by solving $\dualrestrict_\delta^1$, $x$ is a solution to $\dualrestrict^1_\delta(S)$ for some parameter $S \in \Q$, where $f_1(x) \leq \opt_1(S)$ and $f_2(\tilde x) \leq (1+\delta) \cdot S$. Due to this and~(c), we have
		$$f_2(\xnext) \leq \frac 1 {1+\delta} \cdot f_2(x) \leq \frac 1 {1+\delta} \cdot (1+\delta) \cdot S = S,$$
		so we can conclude that $f_1(x) \leq \opt_1(S) \leq f_1(\xnext)$.	 
		\item Consider an iteration of the inner while loop and suppose that (a)--(f) hold at the beginning of this iteration. We show that~(f) also holds at the end of this iteration (provided that the algorithm does not terminate during the iteration). At the beginning of the iteration, we have $f_1(\xnext) = f_1(x)$, so~(d) implies that $\xnext$ strictly dominates~$x$. Therefore, also $P \cup \{\xnext\}$ approximates all solutions~$x' \in X$ with $f_2(x') \geq \frac 1 {1+\varepsilon} \cdot f_2(x)$. Moreover, using~(c) and~(e), we know that
		$$f_1(\xnext) \leq \opt_1(S_2) = \opt_1(\frac 1 {(1+\delta)^2}\cdot f_2(x)) \leq \opt_1(\frac 1 {1+\varepsilon} \cdot f_2(x)),$$
		so $\xnext$ also approximates all solutions $x' \in X$ with $\frac 1 {1+\varepsilon} \cdot f_2(x) > f_2(x') \geq \frac 1 {1+\varepsilon} \cdot f_2(\xnext)$. During the iteration of the while loop, line~11 is performed, which implies that~(f) holds at the end of the iteration.
		
		Now consider an iteration of the outer repeat loop and assume that (a)--(f) hold in line~15 of this iteration. We show that~(f) holds in line~10 of the next iteration (given that line~10 is reached). After performing line~16, we know that~$P$ approximates every solution $x' \in X$ with $f_2(x') \geq \frac 1 {1+\varepsilon} \cdot f_2(x) = S_2$. This implies that~(f) holds after line~5 of the next iteration and, thus, also in line~10 of the next iteration.
		
		\item If~(d) holds at the beginning of a fixed iteration of the while loop, then the value of~$f_2(x)$ decreases during the iteration. Therefore,~(g) holding at the beginning of the iteration together with the fact that~$P$ remains unchanged during the while loop imply that~(g) also holds at the end of the iteration.
		
		If~(g) holds in line~15 of an iteration of the repeat loop, then, after line~16, we have that $\min\{f_2(\bar x) : \bar x \in P\} = f_2(x)$, so $S_2 = \frac 1 {1+\varepsilon} \cdot \min\{f_2(\bar x) : \bar x \in P\}$ holds after line~16. After line~5 in the next iteration, we thus have
		\begin{align*}
		f_2(x) \leq (1+\delta) \cdot S_2 = \frac 1 {(1+\delta)^3} \cdot \min\{f_2(\bar x) : \bar x \in P\}.\tag*{\qed}
		\end{align*}
	\end{enumerate}
\end{proof}

The following two results establish a bound on the number of instances of $\dualrestrict^1_\delta$ solved in the algorithm.

\begin{lemma} \label{lem:minstepsize}
	Every time the value of $S_2$ is changed from some value~$S_2^{\textnormal{old}}$ to a new value~$S_2^{\textnormal{new}}$ during the execution of Algorithm~\ref{alg:twoapproxAlgo}, we have
	$$S_2^{\textnormal{new}} \leq \frac{1}{1+\delta} \cdot S_2^{\textnormal{old}}.$$
\end{lemma}
\begin{proof}
	We distinguish three cases: If the value of~$S_2$ is changed from $S_2^{\textnormal{old}}$ to $S_2^{\textnormal{new}}$ in line~7, we have
	$$S_2^{\textnormal{new}} = \frac{1}{(1+\delta)^2} \cdot f_2(x) \leq \frac{1}{(1+\delta)^2} \cdot (1+\delta) \cdot S_2^{\textnormal{old}} = \frac{1}{1+\delta} \cdot  S_2^{\textnormal{old}},$$
	where the inequality is due to line~5.

	If the value of~$S_2$ is changed from $S_2^{\textnormal{old}}$ to $S_2^{\textnormal{new}}$ in line~12, Lemma~\ref{lem:twoapproxlemma}~(b) and~(d) imply that
	$$S_2^{\textnormal{new}} = \frac{1}{(1+\delta)^2} \cdot f_2(x) = \frac{1}{(1+\delta)^2} \cdot f_2(\xnext) \leq \frac{1}{1+\delta} \cdot  S_2^{\textnormal{old}}.$$
	
	If the value of~$S_2$ is changed from $S_2^{\textnormal{old}}$ to $S_2^{\textnormal{new}}$ in line~16, we have $S_2^{\textnormal{new}} = \frac 1 {1+\varepsilon} \cdot f_2(x)$. Since Lemma~\ref{lem:twoapproxlemma}~(b) yields that $S_2^{\textnormal{old}}= \frac 1 {(1+\delta)^2} \cdot f_2(x)$, this implies that
	\begin{align*}
	S_2^{\textnormal{new}} = \frac{1}{(1+\delta)^2} \cdot  S_2^{\textnormal{old}} < \frac{1}{1+\delta} \cdot  S_2^{\textnormal{old}}.\tag*{\qed}
	\end{align*}
\end{proof}

\begin{proposition}\label{prop:runtime_factor_2}
	Algorithm~\ref{alg:twoapproxAlgo} terminates after solving $\mathcal{O}(\frac M \varepsilon)$ many instances of $\dualrestrict^1_\delta$. 
\end{proposition}
\begin{proof}
	As soon as solving an instance of $\dualrestrict^1_\delta$ yields $\textnormal{NO}$, Algorithm~\ref{alg:twoapproxAlgo} terminates.
	Note that $\dualrestrict^1_\delta(S_2)$ is solved once whenever~$S_2$ is assigned a new value. We show that the number of times the value of~$S_2$ is changed is bounded by $\lfloor\frac{2M}{\log(1+\delta)}\rfloor + 1 = \mathcal{O}(\frac M \varepsilon)$. The initial value assigned to~$S_2$ is~$2^M$. $\dualrestrict^1_\delta(S_2)$ yields $\textnormal{NO}$ if $S_2 < 2^{-M}$, so if $\dualrestrict^1_\delta(S_2)$ is solved for some $S_2 < 2^{-M}$, the algorithm is guaranteed to terminate. But this is the case after at most $\lfloor\frac{2M}{\log(1+\delta)}\rfloor + 1$ many changes of the value assigned to~$S_2$ due to Lemma~\ref{lem:minstepsize}.
	\qed
\end{proof}

\noindent The correctness of the algorithm is established by the following proposition.

\begin{proposition}\label{prop:correctness_factor_2}
	Algorithm~\ref{alg:twoapproxAlgo} returns a one-exact $\varepsilon$-Pareto set.
\end{proposition}

\begin{proof}
	Algorithm~\ref{alg:twoapproxAlgo} terminates as soon as solving $\dualrestrict^1_\delta(S_2)$ yields $\textnormal{NO}$ for some~$S_2$. We do a case distinction on where this happens in Algorithm~\ref{alg:twoapproxAlgo} and show that the solutions in the set~$P$ returned by the algorithm $(1,1+\varepsilon)$-approximate all solutions~$x' \in X$ with $f_2(x') \geq S_2$ (i.e., all solutions~$x'\in X$). If $\dualrestrict^1_\delta(S_2)$ yields $\textnormal{NO}$ in line~5 of the first iteration of the repeat loop, this implies that there does not exist any feasible solution at all. If $\dualrestrict^1_\delta(S_2)$ yields $\textnormal{NO}$ in line~5 of some other iteration of the repeat loop, Lemma~\ref{lem:twoapproxlemma}~(f) holds in line~10 of the previous iteration, in which, in particular, lines~15 and~16 are executed. Therefore, the solutions in the set~$P$ returned by Algorithm~\ref{alg:twoapproxAlgo} approximate all solutions~$x' \in X$ with $f_2(x') \geq \frac 1 {1+\varepsilon} \cdot f_2(x) = S_2$. If $\dualrestrict^1_\delta(S_2)$ yields $\textnormal{NO}$ in line~8 or line~13 of some iteration (of the repeat loop or while loop, respectively), we can show that Lemma~\ref{lem:twoapproxlemma}~(f) holds at this moment by using the same argumentation as in the proof of Lemma~\ref{lem:twoapproxlemma}~(f). Therefore and since $x$ is added to $P$ in these two cases, the solutions in the returned set~$P$ approximate all solutions~$x' \in X$ with $f_2(x') \geq \frac 1 {1+\varepsilon} \cdot f_2(x) > \frac 1 {(1+\delta)^2} \cdot f_2(x) = S_2$.
	\qed 
\end{proof}

It remains to bound the cardinality of the returned one-exact $\varepsilon$-Pareto set~$P$. To this end, we first show in Lemma~\ref{lem:mindistance} that any two solutions in~$P$ always differ by at least a factor of $(1+\delta)^3$ with respect to their $f_2$-values.

\begin{lemma}\label{lem:mindistance}
	Let~$P$ be the set returned by Algorithm~\ref{alg:twoapproxAlgo}. For all $x_1,x_2 \in P$ with $x_1 \neq x_2$ and $f_2(x_1) \geq f_2(x_2)$, we have
	$$f_2(x_1) \geq (1+\delta)^3 \cdot f_2(x_2),$$
	where $(1+\delta)^4 = 1+\varepsilon$.
\end{lemma}
\begin{proof}
	Note that Lemma~\ref{lem:twoapproxlemma}~(g) holds each time a solution is added to~$P$ in line~15 of Algorithm~\ref{alg:twoapproxAlgo}. If some solution~$x \in X$ is added to~$P$ in line~9 or line~14, $\dualrestrict^1_\delta(S_2)$ must have yielded $\textnormal{NO}$ in line~8 or line~13, respectively. We can show that Lemma~\ref{lem:twoapproxlemma}~(g) holds at this moment by using the same argumentation as in the proof of Lemma~\ref{lem:twoapproxlemma}~(g).
	\qed
\end{proof}

Lemma~\ref{lem:solutionquality} establishes that solutions~$x \in P$ are almost efficient in the sense that no other solution with the same or a better $f_1$-value is better by a factor of $(1+\delta)^2$ or more in $f_2$. 

\begin{lemma} \label{lem:solutionquality}
	Let~$P$ be the set returned by Algorithm~\ref{alg:twoapproxAlgo}. For any $x \in P$, there does not exist any feasible solution~$x' \in X$ with $f_1(x') \leq f_1(x)$ and $f_2(x') \leq \frac 1 {(1+\delta)^2} \cdot f_2(x)$.
\end{lemma}
\begin{proof}
	Given $x \in P$, consider the case that~$x$ is added to~$P$ in line~15 of Algorithm~\ref{alg:twoapproxAlgo}. At this point in the algorithm, we know that $f_1(\xnext) \neq f_1(x)$. Thus, Lemma~\ref{lem:twoapproxlemma}~(b), (c), and~(e) imply that $f_1(x) < f_1(\xnext) \leq \opt_1(S_2) = \opt_1(\frac 1 {(1+\delta^2)} \cdot f_2(x))$, so any solution~$x' \in X$ with $f_2(x') \leq \frac 1 {(1+\delta^2)} \cdot f_2(x)$ must have $f_1(x') > f_1(x)$.
	
	Now consider the case that~$x$ is added to~$P$ in line~9 or line~14 of Algorithm~\ref{alg:twoapproxAlgo}. In this case, we have $\dualrestrict_\delta^1(S_2) = \textnormal{NO}$ for $S_2 = \frac 1 {(1+\delta^2)} \cdot f_2(x)$, so there does not exist any solution~$x' \in X$ with $f_2(x') \leq \frac 1 {(1+\delta^2)} \cdot f_2(x)$.
	\qed
\end{proof}

\begin{proposition}\label{prop:size_factor_2}
	Let~$P$ be the set returned by Algorithm~\ref{alg:twoapproxAlgo} and let~$P^*$ be a smallest one-exact $\varepsilon$-Pareto set. Then $$|P| \leq 2 \cdot |P^*|.$$	
\end{proposition}
\begin{proof}
	First, we show that no solution~$x'\in X$ can $(1,1+\varepsilon)$-approximate more than two solutions in the returned set~$P$. 
	Let $x_1,x_2,x_3 \in P$ be three pairwise different solutions in~$P$ such that $f_2(x_1) \geq f_2(x_2) \geq f_2(x_3)$. Let $x' \in X$ be an arbitrary feasible solution that $(1,1+\varepsilon)$-approximates~$x_1$ and~$x_2$. We show that~$x'$ does not $(1,1+\varepsilon)$-approximate~$x_3$. Note that Lemma~\ref{lem:mindistance} implies that
	$$f_2(x_3) \leq \frac 1 {(1+\delta)^3} \cdot f_2(x_2) \leq \frac 1 {(1+\delta)^6} \cdot f_2(x_1).$$
	Since~$x_1$ is $(1,1+\varepsilon)$-approximated by~$x'$, we have $f_1(x') \leq f_1(x_1)$, so Lemma~\ref{lem:solutionquality} implies that
	$$f_2(x') > \frac 1 {(1+\delta)^2} \cdot f_2(x_1).$$
	Combining the two inequalities above yields
	$$f_2(x') > (1+\delta)^4 \cdot f_2(x_3) = (1+\varepsilon)\cdot f_2(x_3),$$
	i.e., $x_3$ is not $(1,1+\varepsilon)$-approximated by~$x'$. An illustration of this is given in Figure~\ref{fig:twoapprox}.
	Note that the above arguments also imply that no solution~$x_0\in P$ with $x_0 \neq x_1$ and $f_2(x_0) \geq f_2(x_1)$ can be approximated by~$x'$. If this was the case, then~$x'$ would not approximate~$x_2$.
	
	Now let~$P^*$ be an arbitrary minimum-cardinality one-exact $\varepsilon$-Pareto set. Then, for any $x \in P$, there exists some $x' \in P^*$ that $(1,1+\varepsilon)$-approximates~$x$. Thus, since any~$x'\in P^*$ can approximate at most two solutions in~$P$, there have to be at least $\left\lceil \frac {|P|} 2 \right\rceil$ many elements in~$P^*$, so $|P| \leq 2 \cdot \left\lceil \frac {|P|} 2 \right\rceil \leq 2 \cdot |P^*|$.\qed
\end{proof}

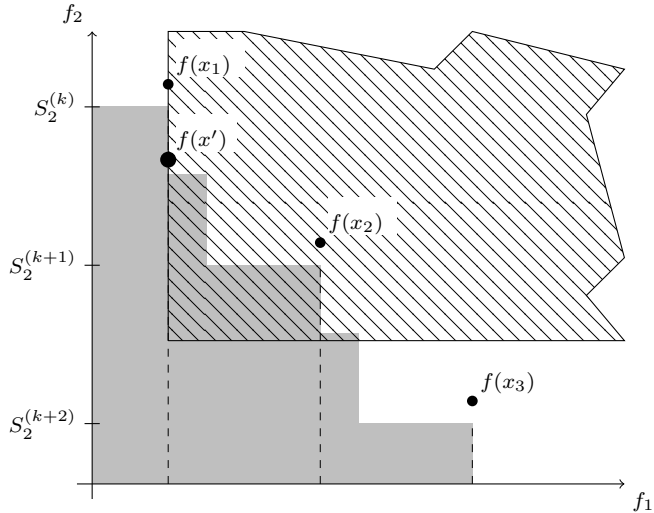
\begin{figure}[ht!]
	\begin{center}
		\begin{tikzpicture}[scale=1.0]
		
		\fill[gray!50] (0,0) rectangle (1,5);
		\fill[gray!50] (1,0) rectangle (1.5,4.1);
		\fill[gray!50] (1.5,0) rectangle (3,2.9);
		\fill[gray!50] (3,0) rectangle (3.5,2.0);
		\fill[gray!50] (3.5,0) rectangle (5,0.8);
		
		\draw[hatchspread = 6 pt, pattern=my north west lines] (1,1.9) -- (1,6) -- (2,6) -- (4.5,5.5) -- (5,6) -- (7,5.5) -- (6.5,4.9) -- (7,3) -- (6.5,2.5) -- (7,1.9) -- (1,1.9);
		
		\fill[white] (1.1,5.4) rectangle (2.0,5.9);
		\fill[white] (3.1,3.3) rectangle (4.0,3.8);
		
		\fill[white] (1.1,4.4) rectangle (1.9,4.9);
		
		\fill (1,5.3) circle (2pt) node[above right] {$f(x_1)$};
		\fill (3,3.2) circle (2pt) node[above right] {$f(x_2)$};
		\fill (5,1.1) circle (2pt) node[above right] {$f(x_3)$};
		
		\fill (1,4.3) circle (3pt) node[above right] {$f(x')$};


		\draw[dashed, line width=0.3pt] (1,0) -- (1,5);
		\draw[dashed, line width=0.3pt] (3,0) -- (3,2.9);
		\draw[dashed, line width=0.3pt] (5,0) -- (5,0.8);	
		
		\draw[->] (-0.2,0) -- (7,0) node[below right] {$f_1$};
		\draw[->] (0,-0.2) -- (0,6) node[above left] {$f_2$}; 
		

		\draw[-] (0.1,5) -- (-0.1,5) node[left] {$S_2^{(k)}$};
		\draw[-] (0.1,2.9) -- (-0.1,2.9) node[left] {$S_2^{(k+1)}$};
		\draw[-] (0.1,0.8) -- (-0.1,0.8) node[left] {$S_2^{(k+2)}$};
		
		\end{tikzpicture}
		\caption{Any feasible solution $x' \in X$ can $(1,1+\varepsilon)$-approximate at most two solutions returned by Algorithm~\ref{alg:twoapproxAlgo}. The solution~$x_i$ is a solution to $\dualrestrict^1_\delta(S_2^{(i)})$ for $i = 1,2,3$, where $(1+\delta)^4 = (1+ \varepsilon)$. The gray area does not contain any feasible point due to Lemma~\ref{lem:solutionquality} and the definition of $\dualrestrict$, so $f(x')$ has to lie outside of this region. The hatched area is the area that is $(1,1+\varepsilon)$-dominated by $f(x')$. \label{fig:twoapprox}}
	\end{center}
\end{figure}

\noindent
Propositions~\ref{prop:runtime_factor_2},~\ref{prop:correctness_factor_2}, and~\ref{prop:size_factor_2} directly yield the following theorem:

\begin{theorem}
	Algorithm~\ref{alg:twoapproxAlgo} computes a one-exact $\varepsilon$-Pareto set~$P$ of cardinality $|P| \leq 2 \cdot |P^*|$ solving $\mathcal{O}(\frac M \varepsilon)$ many instances of $\dualrestrict_\delta^1$, where $(1+\delta)^4 = 1+ \varepsilon$.\qed
\end{theorem}

\subsection{Available Efficient Routine for \restrict}
Diakonikolas and Yannakakis~\cite{Diakonikolas+Yannakakis:approx-pareto-sets} show that, for biobjective optimization problems, the subproblems $\dualrestrict^1$ 
and $\restrict^2$ 
are polynomially equivalent: An answer for an instance of $\dualrestrict^1_\delta$ can be found by solving $\mathcal{O}(M)$~instances of $\restrict^2_\delta$ in a binary search and vice versa. 

They also give an algorithm that computes an $\varepsilon$-Pareto set~$P_\varepsilon$ whose cardinality is not larger than twice the cardinality of a smallest $\varepsilon$-Pareto set~$P_\varepsilon^*$. This algorithm is based on routines for both of these subproblems. In order to compute $P_\varepsilon$, the algorithm solves $\mathcal{O}(\frac M \varepsilon)$ instances of $\dualrestrict^1_\delta$ as well as $\mathcal{O}(\frac M \varepsilon)$ instances of $\restrict^2_\delta$, where $(1+\delta)^3 = 1+\varepsilon$. This a-priori bound can be refined a posteriori to the output-sensitive bound of $\mathcal{O}(|P_\varepsilon|) = \mathcal{O}(|P^*_\varepsilon|)$ many instances of each subproblem that are solved. If only one of the two routines is directly available (such as, e.g., for the min-cost makespan scheduling problem), the polynomial reduction between the two subproblems can be used to solve the other problem in the algorithm by Diakonikolas and Yannakakis, resulting in $\mathcal{O}(\frac {M^2} \varepsilon)$ (or a posteriori $\mathcal{O}(M \cdot |P^*_\varepsilon|)$) solved instances of the subproblem for which a routine is available.

\smallskip

We now show how this algorithm can be slightly modified such that it even computes a one-exact $\varepsilon$-Pareto set in the same a-priori asymptotic running time. The cardinality of the one-exact $\varepsilon$-Pareto set~$Q$ computed by this modified algorithm satisfies $|Q| \leq 2\cdot |P^*|$, where~$P^*$ is a smallest one-exact $\varepsilon$-Pareto set, i.e., the modified algorithm yields the same size guarantee as Algorithm~\ref{alg:twoapproxAlgo}. Since it requires a routine for $\restrict^2$ in addition to a routine for $\dualrestrict^1$, the number of solved subproblems might be in the order of $\frac {M^2} \varepsilon$ (if this routine for $\restrict^2$ consists of simply applying the reduction to $\dualrestrict^1$), which is by a factor of $M$ larger than the number of solved subproblems in Algorithm~\ref{alg:twoapproxAlgo}.

For some common problems such as biobjective $\shortestpath$, however, the best known  algorithms for $\restrict$ and for $\dualrestrict$ have similar asymptotic running times \cite{Horvath+Kis:Dual-RSPP,Ergun+etal:RSPP,Lorenz+Raz:RSPP}. For other problems such as minimum-cost maximum matching, a routine for $\restrict$ is available (see, e.g.,~\cite{Grandoni+etal:new-approaches}), but no specific routine for $\dualrestrict$ is known, i.e., we can only solve $\dualrestrict$ via the reduction to $\restrict$. In these cases, our modification of the algorithm by Diakonikolas and Yannakakis offers an output-sensitive bound on the running time. It solves $\mathcal{O}(|P^*|)$ many subproblems if both routines are available and $\mathcal{O}(M \cdot |P^*|)$ many subproblems if only a routine for $\restrict$ is available. Algorithm~\ref{alg:twoapproxAlgo} solves $\mathcal{O}(\frac M \varepsilon)$ and $\mathcal{O}(\frac {M^2} \varepsilon)$ many subproblems, respectively, in these cases, which is equal to the a-priori bound for the modified algorithm, but might be much larger than the a-posteriori bound for some problems.

The modified algorithm is stated in Algorithm~\ref{alg:modifiedDiakonikolas}. Its proof of correctness is similar to the proof for the original algorithm by Diakonikolas and Yannakakis~\cite{Diakonikolas+Yannakakis:approx-pareto-sets} and is given in the appendix.
\begin{algorithm}[!ht]
	\SetKw{Compute}{compute}
	\SetKw{Break}{break}
	\SetKwInOut{Input}{input}\SetKwInOut{Output}{output}
	\SetKwComment{command}{right mark}{left mark}
	
	\Input{an instance $(X,f)$ of a biobjective minimization problem, $\varepsilon>0$, an algorithm for $\dualrestrict^1$, an algorithm for $\restrict^2$}
	
	\Output{a one-exact $\varepsilon$-Pareto set for $(X,f)$}
	
	\BlankLine
	
	\lIf{$\restrict^2_1(2^M) = \textnormal{NO}$}{halt}
	
	$\delta \leftarrow \sqrt[3]{1+\varepsilon} - 1$
	
	$\xleft \leftarrow \dualrestrict^1_1(2^M)$
	
	$\tilde{x}^{(1)} \leftarrow \restrict^2_\delta(2^M)$
	
	$S_2^{(1)} \leftarrow (1+\delta) \cdot f_2(\tilde{x}^{(1)})$
	
	$x^{(1)} \leftarrow \dualrestrict^1_\delta(S_2^{(1)})$
	
	$B_1^{(1)} \leftarrow f_1(x^{(1)}) - 2^{-2M}$
	
	$Q \leftarrow \{x^{(1)}\}$
	
	$i \leftarrow 1$
	
	\While{$B_1^{(i)} \geq f_1(\xleft)$}{
		
		$\tilde{x}^{(i+1)} \leftarrow \restrict^2_\delta(B_1^{(i)})$
		
		$S_2^{(i+1)} \leftarrow \frac{1+\varepsilon}{1+\delta} \cdot \max\{S_2^{(i)}, \frac 1 {1+\delta} \cdot f_2(\tilde{x}^{(i+1)})\}$
		
		$x^{(i+1)} \leftarrow \dualrestrict^1_\delta(S_2^{(i+1)})$
		
		$B_1^{(i+1)} \leftarrow f_1(x^{(i+1)}) - 2^{-2M}$
		
		$Q \leftarrow Q \cup \{x^{(i+1)}\}$
		
		$i \leftarrow i+1$}
	\Return $Q$
	
	\caption{An alternative $(1,1+\varepsilon)$-approximation for biobjective optimization problems.} \label{alg:modifiedDiakonikolas}
\end{algorithm}


\begin{theorem}\label{thm:modalgo}
	Algorithm~\ref{alg:modifiedDiakonikolas} computes a one-exact $\varepsilon$-Pareto set~$Q$ of cardinality~$|Q| \leq 2 \cdot |P^*|$ solving $\mathcal{O}(|P^*|)$ many instances of $\dualrestrict_\delta^1$ and  $\restrict^2_\delta$ , where $(1+\delta)^3 = 1+ \varepsilon$.
\end{theorem}

Similar to the proof of correctness for the original algorithm for computing $\varepsilon$-Pareto sets by Diakonikolas and Yannakakis~\cite{Diakonikolas+Yannakakis:approx-pareto-sets}, the proof of Theorem~\ref{thm:modalgo} is based on comparing the cardinality of the set computed by Algorithm~\ref{alg:modifiedDiakonikolas} to the cardinality of a smallest one-exact $\varepsilon$-Pareto set. This smallest one-exact $\varepsilon$-Pareto set is assumed to be computed by a greedy procedure similar to the ones given by Diakonikolas and Yannakakis~\cite{Diakonikolas+Yannakakis:approx-pareto-sets}, Koltun and Papadimitriou~\cite{Koltun+Papadimitriou:approx-dom-repr}, and Vassilvitskii and Yannakakis~\cite{Vassilvitskii+Yannakakis:trade-off-curves}. The greedy procedure is based on an (exact) routine for $\constrained$ and is given as Algorithm~\ref{alg:smallestoneexact} in the appendix. The fact that $\constrained^1$ and $\constrained^2$ are polynomially equivalent via the same reduction as for $\dualrestrict^1_\delta$ and $\restrict^2_\delta$ yields the following corollary:

\begin{corollary}\label{cor:constrained}
	For a biobjective optimization problem, it is possible to compute a smallest one-exact $\varepsilon$-Pareto set in fully polynomial time if a polynomial-time algorithm for $\constrained^1$ or $\constrained^2$ is available. 
\end{corollary}

Note that $\constrained$ can, in particular, be solved efficiently if all feasible solutions are given explicitly in the input of a biobjective optimization problem. Another class of problems where $\constrained$ is efficiently solvable are biobjective linear programs. Thus, in both of these cases, a smallest one-exact $\varepsilon$-Pareto set can be found in fully polynomial time. 

\subsection{Impossibility Result for Three or More Objectives}
In contrast to the biobjective case, we now demonstrate that no polynomial-time generic algorithm based on $\dualrestrict$ can produce a constant-factor approximation on the size of a smallest one-exact $\varepsilon$-Pareto set for general problems with more than two objective functions.

\begin{theorem}
	For any $\varepsilon>0$ and any positive integer~$n\in\mathbb{N}_+$, there does not exist an algorithm that computes a one-exact $\varepsilon$-Pareto set~$P$ such that $|P| < n \cdot |P^*|$ for every 3-objective minimization problem and generates feasible solutions only via solving $\dualrestrict_\delta$ for values of~$\delta$ such that $\frac 1 \delta$ is polynomial in the encoding length of the input.
\end{theorem}

\begin{proof}
	Given $\varepsilon>0$ and $n\in\mathbb{N}_+$, we construct two instances~$I_1$ and~$I_2$ with $I_1 = (\{x_0,x_1,\ldots,x_n\},f)$ and $I_2 = (\{x_0,x_1,\ldots,x_n,x'_1,\ldots,x'_n\},f)$, where
	$$f(x_i) = \left(
	\begin{array}{c}
	f_1(x_0) + n - i\\
	(1+\varepsilon)^{2i} \cdot f_2(x_0)\\
	\frac 1 {1+\varepsilon} \cdot f_3(x_0)\\
	\end{array}
	\right) \text{and }
	f(x'_i) = \left(
	\begin{array}{c}
	f_1(x_0) + n - i\\
	(1+\varepsilon)^{2i} \cdot f_2(x_0)\\
	\frac 1 {1+\varepsilon} \cdot f_3(x_0)-1\\
	\end{array}
	\right), i = 1,\ldots,n,$$
	(the values~$f_1(x_0)$, $f_1(x_0)$, and~$f_2(x_0)$ are defined later).
	Then the solution~$x_0$ $(1,1+\varepsilon,1+\varepsilon)$-approximates~$x_i$ for $i=1,\ldots,n$, but $x_0$ does not $(1,1+\varepsilon,1+\varepsilon)$-approximate~$x'_i$ for any~$i$ due to the $f_3$-values. Also, no two solutions from the set $\{x_1,\ldots,x_n,x'_1,\ldots,x'_n\}$ approximate each other except for, possibly, $x_i$ and~$x'_i$ for $i = 1,\ldots, n$.
	Thus, the set~$\{x_0\}$ is a one-exact $\varepsilon$-Pareto set in instance~$I_1$, but any one-exact $\varepsilon$-Pareto set in instance~$I_2$ consists of at least $n+1$~solutions. Hence, an algorithm that computes a one-exact $\varepsilon$-Pareto set~$P$ with $|P| < n \cdot |P^*|$ has to be able to distinguish between~$I_1$ and~$I_2$, i.e., detect the existence of at least one~$x'_i$, using only $\dualrestrict$. 
	
	Following a similar argument as in the proof of Theorem~\ref{thm:generictwo}, one can show that, in order to distinguish between the instances~$I_1$ and~$I_2$, an algorithm has to solve $\dualrestrict_\delta^1$ for some~$\delta$ with $\frac 1 \delta > \frac 1 {1+\varepsilon} \cdot f_2(x_0) - 1$.
	Since~$f_2(x_0)$ can be exponential in the encoding length of the input, the claim follows.
	\qed
\end{proof}

\section{Conclusion and Future Research}
This article addresses the task of computing approximate Pareto sets for multiobjective optimization problems. In particular, we strive for such approximate Pareto sets that are exact in one -- without loss of generality the first -- objective function and obtain an approximation guarantee of $1+\varepsilon$ in all other objectives. We show the existence of such $(1,1+\varepsilon, \ldots, 1+\varepsilon)$-approximate Pareto sets of polynomial cardinality under mild assumptions on the considered multiobjective problem. Our main results address the relation between such a so-called one-exact $\varepsilon$-Pareto set and a singleobjective auxiliary problem, the so-called $\dualrestrict$ problem. Interestingly, this auxiliary problem has been considered in the literature before, but its full potential has not been revealed so far. In fact, we prove equivalence of computing a $(1,1+\varepsilon, \ldots, 1+\varepsilon)$-approximate Pareto set in polynomial time and solving $\dualrestrict$ in polynomial time. This result complements the seminal work of Papadimitriou and Yannakakis~\cite{Papadimitriou+Yannakakis:multicrit-approx}, who characterize the class of problems for which a $(1+\varepsilon, \ldots, 1+\varepsilon)$-approximate Pareto set is polynomial-time computable using the so-called $\gap$ problem. With respect to the approximation quality, one cannot hope for a polynomial-time computable approximate Pareto set that is exact in more than one objective function since this would imply the polynomial-time solvability of a related biobjective problem. In this sense, the factor of $(1,1+\varepsilon, \ldots, 1+\varepsilon)$ obtained here is best possible. Additionally, we provide an algorithm that approximates the cardinality of a smallest one-exact $\varepsilon$-approximate Pareto set by a factor of~$2$ for biobjective problems and show that this factor is best possible. Finally, we demonstrate that, using $\dualrestrict$, it is not possible to obtain any constant-factor approximation on the cardinality for problems with more than two objectives efficiently.

\smallskip
\enlargethispage{1.8\baselineskip}

It should be pointed out that our work provides a general method for computing polynomially-sized one-exact $\varepsilon$-Pareto sets by using an algorithm for the $\dualrestrict$ problem.  If applied to multiobjective $\spanningtree$, our work imposes the first polynomial-time algorithm for computing a one-exact $\varepsilon$-Pareto set and, thus, yields the best possible approximation guarantee for this problem. For multiobjective $\shortestpath$, our general algorithms have running times that are competitive with the running time of the specialized algorithm of~\cite{Tsaggouris+Zaroliagis:mult-shortest-path}. For biobjective $\shortestpath$, our algorithm additionally provides a worst-case guarantee on the cardinality of the computed one-exact $\varepsilon$-Pareto set. Future research could focus on the design of additional problem-specific algorithms that compute one-exact $\varepsilon$-Pareto sets for certain multiobjective (combinatorial) optimization problems with faster running times than the general methods provided here.


%
%


\bibliographystyle{spmpsci}      
\bibliography{Literatur}   

%
%

\clearpage

\section*{Appendix}

    In Algorithm~\ref{alg:smallestoneexact}, we formally state an algorithm that computes a smallest one-exact $\varepsilon$-Pareto set if a subroutine for solving $\constrained$ is given. We prove its correctness in Theorem~\ref{thm:smallestoneexact}. Finally, we present a proof of Theorem~\ref{thm:modalgo}.

\begin{algorithm}[!ht]
	\SetKw{Compute}{compute}
	\SetKw{Break}{break}
	\SetKwInOut{Input}{input}\SetKwInOut{Output}{output}
	\SetKwComment{command}{right mark}{left mark}
	
	\Input{an instance $(X,f)$ of a biobjective minimization problem, $\varepsilon>0$, an algorithm for $\constrained^1$, an algorithm for $\constrained^2$}
	
	\Output{a one-exact $\varepsilon$-Pareto set for $(X,f)$}
	
	\BlankLine
	
	\lIf{$\constrained^2(2^M) = \textnormal{NO}$}{halt}
	
	$\xleft \leftarrow \constrained^1(2^M)$
	
	$\tilde{x}^{(1)} \leftarrow \constrained^2(2^M)$
	
	$B_2^{(1)} \leftarrow (1+\varepsilon) \cdot f_2(\tilde{x}^{(1)})$
	
	$x^{(1)} \leftarrow \constrained^1(B_2^{(1)})$
	
	$B_1^{(1)} \leftarrow f_1(x^{(1)}) - 2^{-2M}$
	
	$P^* \leftarrow \{x^{(1)}\}$
	
	$i \leftarrow 1$
	
	\While{$B_1^{(i)} \geq f_1(\xleft)$}{
		
		$\tilde{x}^{(i+1)} \leftarrow \constrained^2(B_1^{(i)})$
		
		$B_2^{(i+1)} \leftarrow (1+\varepsilon) \cdot f_2(\tilde{x}^{(i+1)})$
		
		$x^{(i+1)} \leftarrow \constrained^1(B_2^{(i+1)})$
		
		$B_1^{(i+1)} 
		\leftarrow f_1(x^{(i+1)}) - 2^{-2M}$
		
		$P^* \leftarrow P^* \cup \{x^{(i+1)}\}$
		
		$i \leftarrow i+1$}
	\Return $P^*$
	
	\caption{Greedy algorithm for computing a smallest one-exact $\varepsilon$-Pareto set  for biobjective optimization problems.} \label{alg:smallestoneexact}
\end{algorithm}

\begin{theorem}\label{thm:smallestoneexact}
	Algorithm~\ref{alg:smallestoneexact} computes a smallest one-exact $\varepsilon$-Pareto set~$P^*$ by solving $\mathcal{O}(|P^*|)$ instances of $\constrained^1$ and of $\constrained^2$.
\end{theorem}
\begin{proof}
	First, note that
	\begin{align*}
	    B_1^{(i+1)} + 2^{-2M} & = f_1(x^{(i+1)}) = \opt_1(B_2^{(i+1)}) \leq \opt_1(f_2(\tilde{x}^{(i+1)})) \\ & \leq f_1(\tilde{x}^{(i+1)}) \leq B_1^{(i)}
	\end{align*}
	for $i = 1,2,\ldots$, where all the steps follow immediately from the algorithm and the definition of $\opt_1(\cdot)$. Thus, the termination condition $B_1^{(i)} < f_1(\xleft)$ is fulfilled after finitely many iterations and Algorithm~\ref{alg:smallestoneexact} returns a set~$P^*$ of finite cardinality. Moreover, we obtain the following statements:
	\begin{enumerate}[(a)]
		\item $x^{(1)}$ $(1,1+\varepsilon)$-approximates any solution~$x' \in X$ with $f_1(x') \geq f_1(x^{(1)})$. This is because in the $f_2$-component, we have
		$$f_2(x^{(1)}) \leq B_2^{(1)} = (1+\varepsilon) \cdot f_2(\tilde{x}^{(1)}) = (1+\varepsilon) \cdot \opt_2(2^M) \leq (1+\varepsilon) \cdot f_2(x').$$
		\item For $i = 2,\ldots,|P^*|$, the solution~$x^{(i)}$ $(1,1+\varepsilon)$-approximates any solution~$x' \in X$ with $f_1(x^{(i-1)}) > f_1(x') \geq f_1(x^{(i)})$ because in the $f_2$-component, we have
		\begin{align*}
		f_2(x^{(i)}) &\leq B_2^{(i)} = (1+\varepsilon) \cdot f_2(\tilde{x}^{(i)}) = (1+\varepsilon) \cdot\opt_2(B_1^{(i-1)})\\
		& \leq (1+\varepsilon) \cdot \opt_2(f_1(x^{(i-1)}))\leq (1+\varepsilon) \cdot f_2(x').
		\end{align*}
		\item There are no solutions~$x' \in X$ with $f_1(x') < f_1(x^{(|P^*|)})$ since otherwise we would have
		$$f_1(x') \leq B_1^{(|P^*|)} < f_1(\xleft) = \opt_1(2^M),$$
		where the strict inequality holds due to the termination condition of the algorithm. 
	\end{enumerate}
	Statements~(a)--(c) imply that the set~$P^*$ computed by Algorithm~\ref{alg:smallestoneexact} is a one-exact $\varepsilon$-Pareto set.
	
	We now show via induction that, for all $k \in \{1,\ldots,|P^*|\}$, there exists a smallest one-exact $\varepsilon$-Pareto set~$P^*_k$ such that $\{x^{(1)},\ldots, x^{(k)}\} \subseteq P^*_k$. This fact for $k = |P^*|$ then yields $|P^*| \leq |P^*_{|P^*|}|$, which completes the proof.
	
	In order to prove the claim for $k = 1$, consider a smallest one-exact $\varepsilon$-Pareto set~$P^*_0$. Let $x_1 \in P^*_0$ be a solution that $(1,1+\varepsilon)$-approximates $\tilde{x}^{(1)}$ and let $x' \in X$ be an arbitrary solution $(1,1+\varepsilon)$-approximated by $x_1$. In the second component, we have $f_2(x_1) \leq (1+\varepsilon) \cdot f_2(\tilde{x}^{(1)}) = B_2^{(1)}$, so, in the first component, we have
	$$f_1(x^{(1)}) = \opt_1(B_2^{(1)}) \leq f_1(x_1) \leq f_1(x').$$
	Now, (a) implies that $x^{(1)}$ also $(1,1+\varepsilon)$-approximates $x'$ and, thus, $x_1$ can be replaced by $x^{(1)}$ in~$P^*_0$.
	
	For the induction step $k \rightarrow k+1$, let $P^*_k$ be a smallest one-exact $\varepsilon$-Pareto set with $\{x^{(1)}, \ldots, x^{(k)}\} \subseteq P^*_k$. Let $x_{k+1} \in P^*_k$ be a solution that $(1,1+\varepsilon)$-approximates $\tilde{x}^{(k+1)}$ and note that $f_1(x_{k+1}) \leq f_1(\tilde{x}^{(k+1)}) \leq B_1^{(k)} < f_1(x^{(k)})$, so we must have $x_{k+1} \in P^*_k \setminus \{x^{(1)}, \ldots, x^{(k)}\}$. Now, let $x' \in X$ be an arbitrary solution that is $(1,1+\varepsilon)$-approximated by $x_{k+1}$ but not by any other solution in the set~$P^*_k$. In the first component, we then have
	$$f_1(x^{(k+1)}) = \opt_1(B_2^{(k+1)}) \leq f_1(x_{k+1}) \leq f_1(x')< f_1(x^{(k)}),$$
	where the strict inequality holds due to~(a) and~(b) and because~$x'$ is not $(1,1+\varepsilon)$-approximated by $x^{(1)}, \ldots, x^{(k)}$.
	Now,~(b) implies that $x^{(k+1)}$ also $(1,1+\varepsilon)$-approximates $x'$ and, thus, $x_{k+1}$ can be replaced by $x^{(k+1)}$ in~$P^*_k$.
	\qed
\end{proof}

\noindent	We now present the proof of Theorem~\ref{thm:modalgo}.
	\begin{proof}[of Theorem~\ref{thm:modalgo}]
	The proof is similar to the proof of the non-modified algorithm by Diakonikolas and Yannakakis~\cite{Diakonikolas+Yannakakis:approx-pareto-sets} and to the proof of Theorem~\ref{thm:smallestoneexact}. In fact, the first part of the proof, where it is shown that Algorithm~\ref{alg:modifiedDiakonikolas} correctly computes a one-exact $\varepsilon$-Pareto set~$Q$ follows exactly the same steps as the proof of Theorem~\ref{thm:smallestoneexact}, so we omit this part here.

	We prove the bound on the cardinality of~$Q$ by comparing~$Q$ to a smallest one-exact $\varepsilon$-Pareto set $P^* = \{x_*^{(1)}, \ldots, x_*^{(|P^*|)}\}$ computed by Algorithm~\ref{alg:smallestoneexact}.
	
	We prove that the following statement holds for all $k \in \N$ by induction on~$k$ (which immediately implies the claim): If $|Q| \geq 2k-1$, then $|P^*| \geq k$, and if $|Q| \geq 2k$, then $f_1(x_*^{(k)}) \geq f_1(x^{(2k)})$.
	
	For $k = 1$, it suffices to show that, if $Q \geq 2$, we have $f_1(x_*^{(1)}) \geq f_1(x^{(2)})$. In order to see this, note that	
	\begin{align*}
	f_2(x_*^{(1)}) &\leq (1+\varepsilon) \cdot \opt_2(2^M)
	\leq (1+\varepsilon) \cdot f_2(\tilde{x}^{(1)})
	= \frac {1+\varepsilon} {1+\delta} \cdot S_2^{(1)} \leq S_2^{(2)}.
	\end{align*}
	This implies that $f_1(x_*^{(1)}) \geq \opt_1(S_2^{(2)}) \geq f_1(x^{(2)})$.
	
	Now suppose that, for some $k \geq 2$, we have $|Q| \geq 2k-2$ and $f_1(x_*^{(k-1)}) \geq f_1(x^{(2k-2)})$. We first show that if $|Q| \geq 2k-1$, then $|P^*| \geq k$.
	Recall that the sequence $(f_1(x_*^{(1)}), \ldots, f_1(x_*^{(|P^*|)}))$ is strictly decreasing. Since the solution~$\xleft$ must be $(1,1+\varepsilon)$-approximated by some solution in $P^*$, we must have $f_1(x_*^{(|P^*|)}) \leq f_1(\xleft)$, but for $x_*^{(k-1)}$, we have
	$$f_1(x_*^{(k-1)}) \geq f_1(x^{(2k-2)}) > B_1^{(2k-2)} \geq f_1(\xleft),$$
	so $|P^*| > k-1$.
	
	Finally, we show that if $|Q| \geq 2k$, then $f_1(x_*^{(k)}) \geq f_1(x^{(2k)})$.	Note that we have $f_1(\tilde{x}^{(2k-1)}) \leq B_1^{(2k-2)} < f_1(x^{(2k-2)}) \leq f_1(x_*^{(k-1)})$ due to the induction hypothesis, so $\tilde{x}^{(2k-1)}$ is not $(1,1+\varepsilon)$-approximated by $x_*^{(1)},\ldots,x_*^{(k-1)}$. Let $x_*^{(i)} \in P^*$ be a solution that $(1,1+\varepsilon)$-approximates $\tilde{x}^{(2k-1)}$. In the second component, we then have
	$$f_2(x_*^{(i)}) \leq (1+\varepsilon) \cdot f_2(\tilde{x}^{(2k-1)}) \leq \frac {1+\varepsilon} {1+\delta} \cdot S_2^{(2k-1)} \leq S_2^{(2k)}.$$
	The above argument implies $i \geq k$, so
	$$f_1(x_*^{(k)}) \geq f_1(x_*^{(i)}) \geq \opt_1(S_2^{(2k)}) \geq f_1(x^{(2k)}),$$
	which finishes the proof.
	\qed
\end{proof}

\end{document}